\documentclass[acmsmall,screen,nonacm]{acmart}
\acmJournal{TALG}
\usepackage{latexsym}              
\usepackage{amsmath}               
\usepackage{amsthm}                
\usepackage{color}
\usepackage{enumitem}
\usepackage{comment}
\usepackage{marginnote}

\newcommand{\Rem}{{\rm Rem}}

\newcommand{\mylabel}[1]{\label{#1}}
\newcommand{\ind}[1]{\hspace*{#1em}}

\newcommand{\For}{{\bf for\ }}
\newcommand{\From}{{\bf from\ }}
\newcommand{\Downto}{{\bf downto\ }}
\newcommand{\To}{{\bf to\ }}

\newcommand{\Do}{{\bf do\ }}
\newcommand{\Od}{{\bf od}}

\newcommand{\If}{{\bf if\ }}
\newcommand{\Then}{{\bf then\ }}

\newcommand{\Fi}{{\bf fi}}

\newcommand{\Break}{{\bf break}}

\newcommand{\Return}{{\bf return\ }}

\newcommand{\diag}{{\rm diag}}

\newcommand{\Z}{\ensuremath{\mathbb Z}}
\newcommand{\Q}{\ensuremath{\mathbb Q\mskip1mu}}

\newcommand{\R}{{\mathsf{R}}}

\newcommand{\Znn}{\ensuremath{\mathbb Z}^{n\times n}}

\newtheorem{theorem}{Theorem}

\newtheorem{definition}[theorem]{Definition}
\newtheorem{corollary}[theorem]{Corollary}

\newtheorem{remark}[theorem]{Remark}
\newtheorem{lemma}[theorem]{Lemma}
\newtheorem{example}[theorem]{Example}

\DeclareMathOperator{\rowmod}{{\mathbf r}mod}
\DeclareMathOperator{\colmod}{{\mathbf c}mod}
\DeclareMathOperator{\loglog}{loglog}

\DeclareMathOperator{\Span}{\rm Span}
\newlength{\algwidth}
\usepackage{float}

\title{A Cubic Algorithm for Computing the Hermite Normal Form of a
Nonsingular Integer Matrix}

\author{Stavros Birmpilis}
\orcid{0000-0002-1766-9814}
\email{sbirmpil@uwaterloo.ca}
\author{George Labahn}
\orcid{0009-0005-8977-4890}

\email{glabahn@uwaterloo.ca}
\author{Arne Storjohann}
\orcid{0000-0002-6354-8810}
\email{astorjoh@uwaterloo.ca}
\affiliation{
  \institution{University of Waterloo}
  \department{Cheriton School of Computer Science}
  \streetaddress{200 University Ave W}
  \city{Waterloo}
  \postcode{N2L 3G1}
  \country{Canada}
}



\setcitestyle{nosort}
\begin{document}

\begin{abstract}
A Las Vegas randomized algorithm is given to compute the Hermite
normal form of a nonsingular integer matrix $A$ of dimension $n$.
The algorithm uses quadratic integer  multiplication and cubic
matrix multiplication and has running time bounded by $O(n^3 (\log
n + \log ||A||)^2(\log n)^2)$ bit operations, where $||A||= \max_{ij}
|A_{ij}|$ denotes the largest entry of $A$ in absolute value. A
variant of the algorithm that uses pseudo-linear integer multiplication
is given that has running time $(n^3 \log ||A||)^{1+o(1)}$ bit
operations, where the exponent $``+o(1)"$ captures additional factors
$c_1 (\log n)^{c_2} (\loglog ||A||)^{c_3}$ for positive real constants
$c_1,c_2,c_3$.
\end{abstract}

\maketitle


\section{Introduction}

Corresponding to any nonsingular integer matrix  $A \in \Z^{n \times
n}$, there is a unimodular matrix $U \in \Z^{n \times n}$ such that
\begin{equation}\label{intro:hermite}
H =  U A = \left [ \begin{array}{cccc}
h_1 & h_{12} & \cdots & h_{1n} \\
& h_2 & \cdots & h_{2n} \\
& & \ddots & \vdots \\
& & & h_n
\end{array} \right ] \nonumber
\end{equation}
has all entries nonnegative, and off-diagonal entries $h_{\ast j}$
strictly smaller than the diagonal entry $h_j$ in the same column.
$H$ is the (integer) Hermite normal form of $A$.  The form is unique
with its existence dating back to \citet{Hermite}.  The rows of $H$
give a canonical basis for the lattice generated by the $\Z$-linear
combinations of the rows of $A$. In addition to being upper triangular
and canonical, an important property of the basis given by the
Hermite form is that it requires only $O(n^2(\log n + \log ||A||))$
bits to represent, compared to $O(n^2\log ||A||)$ to write down the
input matrix.

Applications of the Hermite form are well known including, for
example, solving systems of linear diophantine equations
\citep{ChouCollins}, integer programming \citep{Schrijver}, and
determining rational invariants and rewriting rules of scaling
invariants~\citep{HubertLabahn2}, to name just a few.

Algorithms for computing Hermite normal forms for integer matrices
were initially based on triangularizing the input matrix using
variations of Gaussian elimination that used the extended Euclidean
algorithm to eliminate entries below the diagonal. However, such
methods can be prone to exponential expression swell, that is, the
problem of rapid growth of intermediate integer operands.  The first
provably polynomial time algorithm was given by \citet{KannanBachem},
with \citet{ChouCollins} improving this to a running time of $(n^{6}
\log ||A||)^{1+o(1)}$ bit operations.  \citet{DomichKannanTrotter},
\citet{Iliopoulos89:1} and \citet{HafnerMcCurley89} later improved
these to  $(n^{4} \log ||A||)^{1+o(1)}$.  Further improvements came
from \citet{StorjohannLabahnISSAC96} and \citet{thesis}, with worst
case time complexity bounded by $(n^{\omega+1} \log ||A||)^{1+o(1)}$
bit operations, where $\omega$ is the exponent of matrix multiplication.
The standard algorithm for matrix multiplication  has $\omega=3$,
while the current best known asymptotic upper bound for $\omega$
by \citet{AlmanWilliams2021} allows $\omega < 2.37286$.

Recently, a number of approaches have focused on heuristic methods
to achieve faster computation, for example \citep{MicciancioWarinschi01,
PernetStein10, PauderisStorjohann13, LiuPan19} with the last citation
having a complexity of $(n^\omega \log ||A||)^{1 + o(1)}$ in the
case of random input matrices. However, these algorithms require
strong assumptions, for example, that there be only a small number
of non-trivial ($\neq 1$) late diagonal entries of the Hermite form,
something common with random matrices.

In this paper, we give a new randomized algorithm for computing the
Hermite normal form of a nonsingular integer matrix $A \in \Z^{n
\times n}$.  Assuming the use of standard (quadratic) integer
multiplication and standard (cubic) matrix multiplication, the
algorithm has a worst case running time bounded by $O(n^3 (\log n
+ \log ||A||)^2 (\log n)^2)$ bit operations.  If we use a subcubic
matrix multiplication algorithm, for example Strassen's algorithm,
then the cost is $O(n^3 (\log n + \log ||A||)^2)$.  We also give a
variant of our algorithm that has a complexity of  $(n^3 \log
||A||)^{1+o(1)}$ bit operations, assuming fast (pseudo-linear)
integer multiplication.  In all cases, our Hermite form algorithms
are probabilistic of type Las Vegas. That is, the algorithm can
report {\sc Fail} with probability at most $1/2$ but otherwise
returns an answer that is certified to be correct. The three key
ideas that we use are minimal matrix denominators, Smith massagers
and duality of row Hermite and column Howell forms.

We remark that one can also define the Hermite form for a matrix
of univariate polynomials with coefficients from a field. In this
case, the definition requires that the diagonal elements are monic,
while the off-diagonal entries have lower degree than the diagonal
entry in the same column.  The algorithms mentioned in the third
paragraph of this section all have corresponding versions which
work for the polynomial Hermite form, and have a complexity similar
to the integer based algorithms, but with degree taking the place
of bitlength and counting field operations instead of bit operations.
However, there are new, very efficient algorithms which work in
the polynomial case but which have no counterpart in the integer
case. In particular we mention the recent fast algorithm of
\citet{LabahnNeigerZhou17}. This algorithm is deterministic and
computes the (polynomial) Hermite form with a complexity of
$(n^{\omega} \lceil s \rceil)^{1+o(1)}$ field operations, 
with $s$ being the minimum
of the average of the degrees of the columns of $A$ and that of its
rows.  Unfortunately, some of the tools used in that algorithm do
not have counterparts in the case of integer matrices. In particular,
for polynomial matrices one has notions such as degree shifts
\citet{BeckermannLabahnVillard}, order bases
\citet{BeckermannLabahn94,ZhouLabahn12}, column bases \citet{ZhouLabahn13}
and  minimal nullspace bases \citet{ZhouLabahnStorjohann2012} along
with algorithms for their fast computation.  For example, the fast Hermite
algorithm of \citet{LabahnNeigerZhou17} works by directly triangularizing the input matrix,
but is able to exploit the aforementioned tools, that are particular
to polynomial matrices, in order keep degrees of intermediate polynomials
controlled while at the same time maintaining a good complexity.

The rest of this paper is organized as follows. Section~\ref{sec:2}
gives an overview of our approach.  Sections~\ref{sec:3} and~\ref{sec:4}
introduce the mathematical and computational tools we use, including
minimal denominators, Smith massagers, compact representations of
both Hermite forms and Smith massagers, and some basic subroutines.
Section~\ref{sec:diag} then gives an algorithm for determining the
diagonal elements of the Hermite form.  Section~\ref{ssec:colhowell}
describes the column Howell form of a matrix over $\Z/(s)$ for
positive modulus $s$, while Section~\ref{sec:dual} relates the
column Howell form to the inverse of the Hermite form.  Section~\ref{sec:8}
then shows how we compute the Hermite form from a Howell form
corresponding to the inverse of the Hermite form, with
Section~\ref{sec:comphowtrans} detailing our modification of Howell's
algorithm to compute a transformation matrix to produce the required
Howell form.  Section~\ref{sec:dotprod} gives an algorithm to compute
a type of scaled matrix vector product which is essential to obtaining
the running time bound of our algorithm.  Section~\ref{sec:11} uses
the results of the previous section to obtain our main result: a
Las Vegas algorithm for the Hermite form with expected running time
$O(n^3(\log n + \log ||A||))^2(\log n)^2)$ bit operations assuming
standard integer and matrix multiplication.  Section~\ref{sec:fhermalg}
gives a variant of the algorithm that has running time $(n^3 \log
||A||)^{1+o(1)}$ bit operations assuming fast (pseudo-linear) integer
multiplication.  The final section gives a conclusion along with
some topics for future research.

\subsection*{Cost model}

The number of bits in the binary representation of an integer $a$
is given by $$ \lg a = \left\{ \begin{array}{ll} 1 & \mbox{if $a=0$}
\\ 1+ \lfloor \log_2 |a| \rfloor & \mbox{if $a>0$} \end{array}
\right. $$ Using standard integer arithmetic, $a$ and $b$ can be
multiplied in $O((\lg a) (\lg b))$ bit operations, and we can express
$a = qb + r$, with $0 \leq |r| < |b|$, in $O((\lg a/b) (\lg b))$
bit operations.  This complexity model was popularized by \citet{Collins}
and is sometimes called ``naive bit complexity'' (see, for example,
\citet{BachShallit}).

For an integer vector $v$, it will be convenient to define the
\emph{bitlength of $v$} to mean the bitlength of the largest entry
of $v$ in absolute value.

\section{Our approach}\label{sec:2}
 
In this section, we give a high level description of our approach
for computing the Hermite form $H\in\Znn$ of a nonsingular input
matrix $A \in \Z^{n \times n}$.  As previously mentioned, there is
a unimodular matrix $U \in \Z^{n \times n}$ such that $H=UA$.
Multiplying both sides of this equation on the right by $A^{-1}$
gives
\begin{equation} \mylabel{eq:new}
HA^{-1} = U.
\end{equation}
The basis of our approach is to recast the problem of computing
$H$, a unimodular row triangularization of $A$, into that of finding
a {\em minimal left denominator} of $A^{-1}$.  It follows from the
uniqueness of the Hermite form that $H$ can be defined to be the
matrix (in Hermite form) that clears the denominators of $A^{-1}$
under premultiplication and has minimal determinant (i.e., $\det H
= |\det A|$, since $\det U = \pm 1$).

To avoid working with fractions, define $A^* = sA^{-1}$, where $s
\in \Z_{>0}$ is minimal such that $sA^{-1}$ is integral.  Then~$H
A^{-1}  \in \Z^{n \times n}$ holds if and only if
$$
HA^* =  0_{n \times n} \bmod s.
$$
Unfortunately, $A^*$ requires $\Omega(n^3 (\log n + \log ||A||))$
bits to write down in the worst case, and by working with $A^*$
explicitly we do not know how to achieve our target complexity.
However, this approach allows us to bring the Smith form of $A$
into play and reduce the space requirements.

Let $S=\diag(s_1,\ldots,s_n =: s)$ be the Smith form of $A$, and
let $V,W \in\Z^{n \times n}$ be unimodular matrices satisfying $A
V = W S$. Then,
\[A^* \equiv_R V S^*\]
where $S^* = s S^{-1} \in \Z^{n \times n}$ and $\equiv_R$ denotes
right equivalence by unimodular matrices over $\Z$. Such an equivalence also
holds modulo $s$ for a matrix $M = \colmod(V, S)$. Here, $\colmod$
denotes working modulo columns: column $j$ of $M$ is equal to column
$j$ of $V$ reduced modulo $s_j$, $1\leq j \leq n$.  The matrix $M$
is called a {\em reduced Smith massager} of $A$.  The fact that
\begin{equation}\label{mult}
A^* \equiv_R M S^*\bmod s 
\end{equation}
then implies that $A^{-1}$ and $MS^{-1}$ have the same  minimal
left denominator in Hermite form, namely, for any $H \in \Z^{n
\times n}$, we have $HA^* = 0_{n \times n} \bmod s$ if and only if
$$HMS^* = 0_{n \times n} \bmod s.$$ This allows us to look for a
minimal left denominator in Hermite form for a matrix with total
size controlled by the Smith form $S$: the space required to store
$M$ is $O(n^2(\log n + \log ||A||))$ bits. Moreover, there is an
existing algorithm that can compute both $S$ and $M$ quickly.

The special form of the matrix $M S^{-1}$ and the uniqueness of
Hermite forms has a number of advantages for efficient computation.
First, by using an algorithm of \citet{PauderisStorjohann13}, we can
find a minimal triangular denominator for $M S^{-1}$, expressed as
a product of $n$ minimal Hermite denominators. While this does not
produce the Hermite form $H$ of $A$, the product of the diagonals
of these $n$ triangular matrices gives the diagonal entries of $H$.
We show that the overall cost of obtaining the diagonal entries of
$H$ from $M$ and $S$ is $O(n (\log \det S)^2)$ bit operations.  This
allows us to overcome one of the biggest issues in designing a fast
algorithm for the Hermite form in the worst case, that is, we now
know the bitlength of each of the columns of $H$.

Notice that finding $H^{-1}$ is equivalent to finding the Hermite
form, since $H$ is triangular. Indeed, let $H_j$ be equal to $I_n$
except with column $j$ equal to that of $H$, $1\leq j\leq n$. Then,
since both $H$ and $H^{-1}$ are upper triangular, 
there is a simple iterative scheme to go from $H^{-1}$ to $H$ shown
in Figure~\ref{figure1}. We remark that in the first line of the $j$-loop
in Figure~\ref{figure1}, the principal leading $(j-1) \times (j-1)$ submatrix of $\bar{H}$ will be $I_{j-1}$, and column $j$ of $\bar{H}$ will have the form
$$
-\frac{1}{h_{j}} \left [ \begin{array}{c} h_{1j} \\
 \vdots \\
 h_{j-1,j}\\
 -1 \\
\phantom{0} \\
\phantom{0} \\
\end{array} \right],
$$
from which $H_j$ is easily recovered.
\begin{figure}[H] \centering
\fbox{
\begin{minipage}{.5\textwidth}
$\bar{H} := H^{-1}$\\
\For $j=1$ \To $n$ \Do\\
\ind{1} Recover $H_j $ from column $j$ of $\bar{H}$\\
\ind{1} $\bar{H} := H_j \bar{H}$\\
\Od\\
\Return $H_n H_{n-1} \cdots H_1$
\end{minipage}
}
\caption{\label{figure1} Hermite form $H$ from $H^{-1}$}
\end{figure}

However, we can do better. The same process can work without having
$H^{-1}$ exactly. Since there exists a unimodular matrix $U$ such
that $U A = H$, then by letting $H^*=  s H^{-1}$, we can write this
as the dual problem
$$
A^* U^* = H^*
$$
with $U^*$ unimodular. Since $H^*$ is an upper triangular integer
matrix, we later show that we can replace $H^*$ by any upper
triangular matrix having the same diagonal entries and which is
right equivalent to $H^*$ modulo $s$. The natural form for such a
matrix is the column Howell form $T$, a type of column reduced
echelon matrix over the residue class ring $\Z/(s)$.

This implies that we can construct the Hermite form from any column
Howell form $T$ that is right equivalent to $H^{\ast}$ over $\Z/(s)$.
This allows us to replace $H^{-1}$ by $T$ in the procedure shown
in Figure~1, and to work modulo $s$, and thus avoid explicit
fractions.


\begin{example} \label{exbasis}
Let
$$
A = \left[ \begin {array}{cccc} -13&10&-20&27\\ {}27&30&
15&30\\ {}0&15&15&6\\ {}-21&0&-15&9
\end {array} \right]. 
$$
\cite[Algorithm \texttt{SmithMassager}]{BirmpilisLabahnStorjohann20}
gives the Smith form $S = \diag(s_1,s_2,s_3,s_4 =: s) =  \diag(1,3,15,105 =:s)$ 
and a Smith massager $M$ for $A$ as
$$
M := \left[ \begin {array}{cccc} 0&2&0&55\\ {}0&0&7&32
\\ {}0&2&2&41\\ {}0&2&10&10
\end {array} \right].
$$
Let $S^* = s S^{-1}$.  By computing a minimal
denominator of $M$ that is expressed as the product of four
upper triangular matrices, we determine the diagonal elements 
of $H$ to be $h_1,h_2,h_3,h_4 = 1, 15, 15, 21$.
A Howell form of 
$M S^*  \in \Z/(s)^{n \times n}$ 
with the appropriate diagonal elements of $H^*$ is then given by 
\[
T  =    \left[ \begin {array}{cccc} 105&70&70&45\\ &7&0&
100\\ &&7&101\\ &&&5
\end {array} \right]
=
\left [ \begin{array}{cccc} \frac{s}{h_{1}}  & 70 & 70 & 45 \\
 & \frac{s}{h_{2}}  & 7 & 100 \\
 & & \frac{s}{h_{3}}  & 101 \\
 & & & \frac{s}{h_{4}} \end{array} \right ] .
 \]
Section~\ref{sec:dual} shows that column $j$ of $(H_{j-1}\cdots H_1)T$
is congruent modulo $s$ to
$$
-\frac{s}{h_{j}} \left [ \begin{array}{c} h_{1j} \\
 \vdots \\
 h_{j-1,j}\\
 -1 \\
\phantom{0} \\
\phantom{0} \\
\end{array} \right] \bmod s,
$$
from which $H_j$ is easily recovered.
Using $T$ instead of $H^{-1}$ in the procedure of Figure~\ref{figure1}
and working modulo $105$ then gives
\begin{eqnarray*}
j=1: & & 
H_1 = \left [ \begin{array}{cccc}
1 & \phantom{0} & \phantom{0} & \phantom{0} \\
 & 1 & & \\ 
& & 1 & \\
& & &  1 \end{array} \right]
\mbox{~~~and~~~}
H_1 T =  T \\
j=2: & & H_2 = \left [ \begin{array}{cccc}
1 & 5 & \phantom{0} & \phantom{0} \\
 & 15 & & \\
& & 1 & \\
 & & & 1 \end{array} \right ]
\mbox{~~and~~}
 H_2 H_1 T = \left[ \begin {array}{cccc} 
 \phantom{0}&\phantom{0}&70&20\\ 
 &&0&30\\
 &&7&101\\
 && &5\end {array} \right] \\
j=3: & & H_3  = \left [ \begin{array}{cccc}
1 & \phantom{0} & 5 & \phantom{0} \\
 & 1 &0 & \\
& & 15 & \\
 & & & 1 \end{array} \right ]
 \mbox{~~and~~}
H_3 H_2 H_1  T =  \left[ \begin {array}{cccc} 
\phantom{0}&\phantom{0}&\phantom{0}&0\\
&&&30\\
 &&&45\\
&&&5\end {array}
 \right]  \\
j=4: & & H_4  = \left [ \begin{array}{cccc}
1 & \phantom{0} &   & 0 \\
 &  1  &  &15 \\
& & 1 & 12 \\
 & & & 21 \end{array} \right ] 
 \mbox{~~and~~}
H_4 H_3 H_2 H_1  T  = 0_{4 \times 4}
\end{eqnarray*}
with the Hermite basis given by 
$$ H_4 H_3 H_2 H_1 = 
\left[ \begin {array}{cccc} 1&5&5&0\\ &15&0&15
\\ &&15&12\\ &&&21
\end {array} \right].
$$
\end{example}

Unfortunately, as mentioned previously for $A^*$, the size of a
Howell form $T$ can be $\Omega(n^3 (\log n + \log ||A||))$ bits in
the worst case, and by working with $T$ directly we do not know how
to achieve our target complexity.  Instead, we compute a matrix
$\tilde{U}$ satisfying
\begin{equation}\label{howell}
T =  M S^* \tilde{U} \bmod s,
\end{equation}
where $S^* = s S^{-1}$. Furthermore, in the same way that we could
assume that $M$ was column reduced modulo $S$, we may assume that
$\tilde{U}$ is row reduced modulo $S$.  The number of bits required
to represent all three matrices on the right hand side of (\ref{howell})
is then $O(n^2(\log n  + \log ||A||))$.

The matrix $\tilde{U}$ can be found by a simple modification of
Howell's original algorithm for determining his normal form. In
order to then find column $j$ of $H_{j-1}H_{j-2} \cdots H_1T$, we
need to determine $$(v_1, \ldots, v_n) = (- h_{1j}, \ldots ,
-h_{j-1,j}, 1, 0, \ldots , 0)$$  satisfying the equation
$$
\frac{s}{h_j} 
\overbrace{
\left [ \begin{array}{c} v_1\\ 
 \vdots \\
v_n
\end{array} \right]}^{\textstyle v}  \equiv
\overbrace{
\left [\begin{array}{ccc}
m_{11} & \cdots & m_{1n}\\   
\vdots &  \ddots & \vdots\\  
m_{n1} & \cdots & m_{nn} 
 \end{array}
 \right ]}^{\textstyle \tilde{M}}
\overbrace{
\left [ \begin{array}{ccc}
\frac{s}{s_1} & & \\
  & \ddots & \\
  & & \frac{s}{s_n} \end{array} \right ]}^{\textstyle S^*}
\overbrace{
\left [ \begin{array}{c} u_{1} \\
\vdots \\
u_{n} \end{array} \right ]}^{\textstyle u}
\bmod s,
$$
where $\tilde{M} = \colmod(H_{j-1}H_{j-2} \cdots H_1M,S)$, and $u$
is column $j$ of $\tilde{U}$.  To compute this matrix vector product
with the intermediate scaling matrix $S^*$, we take advantage of
the fact that $\tilde{M}$ and $u$ are column and row reduced modulo
$S$, respectively.  We also exploit the fact that we have precomputed
the diagonal entries of the Hermite form, and thus know the scaling
factor $s/h_j$.  This allows us to achieve a cost estimate for computing
column $j$ that depends on $\log ||v|| \leq \log h_j$ instead of
$\log s$.

Ultimately, our algorithm computes the Hermite form in a column by
column basis, with the  computation for column $j$ requiring
$$O(n^2(\log n + \log ||A||)(\log h_{j} + \log n + \log ||A||))$$
bit operations.  Adding over all iterations $1\leq j\leq n$ then
gives the total cost of our algorithm.

\section{Mathematical preliminaries} \mylabel{sec:3}

In this section, we discuss some basic mathematical building blocks
used in our Hermite form algorithm. These include minimal denominators
of rational matrices, Smith massagers of $A$, and data structures
for the compact representation of Hermite forms and Smith massagers.

\subsection{Minimal denominators} \mylabel{subsec:2.1}

\begin{definition}  \label{def:minden}
A \emph{(left) denominator} of a matrix $B \in \Q^{n \times m}$ is
a matrix $H \in \Z^{n \times n}$ whose rows are  in the lattice
\begin{equation} \mylabel{lat} \{v \in \Z^{1 \times n} \mid vB \in
\Z^{1 \times m} \}.  \end{equation} $H$ is a \emph{minimal denominator}
if the rows of $H$ are a basis for~(\ref{lat}).  The \emph{minimal
Hermite denominator} is the unique minimal denominator that is in
Hermite form.
\end{definition}

For example, a minimal denominator of a zero matrix with $n$ rows
is $I_n$, while a minimal denominator of $A^{-1}$ is $A$ itself.
The minimal Hermite denominator of $A^{-1}$ is $H$, the Hermite
form of $A$. Similarly, if $A^{-1}$ and $B^{-1}$ are right equivalent
then they have the same minimal Hermite denominator.
\begin{example} \mylabel{ex:nontriv}
The minimal Hermite denominator of 
$$
\frac{1}{16}\left [\begin{array}{c} 1 \\
 4 \\
 4\\
8 \end{array} \right ] \in \Q^{4 \times 1}
$$
is 
$$
H = \left [ \begin{array}{cccc} 4 & 0 & 1 & 1 \\
 & 1 & 1 & 1 \\
 & & 2 & 1 \\
 & & & 2 \end{array} \right ] \in \Z^{4 \times 4}.
$$
This shows that a rational matrix with $n$ rows but with fewer than $n$ 
columns can encode a nontrivial $n \times n$ Hermite form.
\end{example}

The next  two lemmas follow from the fact that a minimal denominator
is a basis for the lattice shown in~(\ref{lat}).
\begin{lemma} \mylabel{denequiv} 
Any two minimal denominators for a $B \in \Q^{n \times m}$ are left
equivalent over $\Z$.
\end{lemma}
\begin{lemma}
The determinant of a minimal denominator for a $B \in \Q^{n \times m}$ 
divides the determinant of any other denominator of $B$.
\end{lemma}

Important for our work is that minimal denominators can be computed
in parts as shown by the following lemma.
\begin{lemma} \mylabel{lem:recurse}
Decompose $B \in \Q^{n \times m}$ arbitrarily 
as $B = \left [ \begin{array}{c|c} B_1 & B_2 \end{array} \right ]$.
If $H_1$ is a minimal denominator
of $B_1$, and $H_2$ is a minimal denominator of $H_1 B_2$, then $H_2 H_1$
is a minimal denominator of $B$.
\end{lemma}

\begin{proof}
It is evident that $H_2H_1$ is a denominator of $B$, and hence, we
only need to show that it is minimal. If it is not a minimal
denominator, then there exist matrices $H, W\in\Znn$ such that $H$
is a minimal denominator, $H_2 H_1 = W H$ and $W$ is not unimodular.

However, since $H_2=WHH_1^{-1}$ is a minimal denominator of $H_1B_2$,
then $WH$ must be a minimal denominator of $B_2$. This is a
contradiction since $H$ is a denominator of $B_2$ and $W$ is not
unimodular.
\end{proof}

Finally, recall that any rational number can be written as an integer
and a proper fraction.  For example,
\begin{equation} \label{eq:frac}
{\frac{9622976468279041913}{21341}}  = 450914974381661 + {\frac{14512}{21341}},
\end{equation}
where $450914974381661$ is the quotient and 14512 is the remainder
of the numerator with respect to the denominator.
We see that, for any rational matrix $B$, if $s$ is a positive integer
such that $sB$ is integral, then
the proper fraction $\Rem(sB,s)/s$ and $B$ have the same
denominators.  Here, $\Rem$ denotes the positive remainder.
Thus, instead of working with the rational matrix $B$, 
we can work with the matrix $\Rem(sB,s)$ over $\Z/(s) = \{0,1,\ldots,s-1\}$.
\begin{lemma} \mylabel{lem:1}
For  $B \in \Q^{n \times m}$ and any $s \in \Z_{>0}$ such that $sB$
is integral, we have: $\{v \in \Z^{1 \times n} \mid  vB \in \Z^{1
\times m} \} = \{ v \in \Z^{1 \times n} \mid v(sB) \equiv 0_{1
\times m} \bmod s \}$.
\end{lemma}

\begin{remark} \mylabel{rem:U}
If $U \in \Z/(s)^{m \times m}$ satisfies $\det U \perp s$, then $H$ 
is a (minimal) denominator of $B$ if and only if $H$ is a (minimal)
denominator of $BU$. Here, $\perp$ denotes two integers being relatively prime. 
\end{remark}

\subsection{Smith massagers}

Important for our work is the notion of a Smith massager of $A$.

\begin{definition}[\protect{\cite[Definition~1]{BirmpilisLabahnStorjohann21}}]
 \label{def:sm}
Let $A\in\Znn$ be a nonsingular integer matrix with Smith form $S$.
A matrix $M\in\Znn$ is a \emph{Smith massager} for $A$ if
\begin{itemize}
\item[(i)] it satisfies that
\begin{equation} \label{eq:AM}
AM\, \equiv\, 0\, \colmod S, and
\end{equation}
\item[(ii)] there exists a matrix $\hat{W}\in\Znn$ such that
\begin{equation}\label{eq:UM}
\hat{W}   M\, \equiv\, I_n\, \colmod S.
\end{equation}
\end{itemize}
\end{definition}
It follows directly from Definition~\ref{def:sm} that if $M$ is a
Smith massager for $A$, then $\colmod(M,S)$ is also a Smith massager
for $A$.  If $M= \colmod(M,S)$, then $M$ is called a \emph{reduced
Smith massager}.  Compared to $A^{-1}$, a reduced Smith massager
$M$ requires only $O(n^2(\log n + \log ||A||))$ space to store.

The key feature of a Smith massager that we exploit in this paper
is the following.

\begin{lemma} \label{lem:key}
Let $A \in \Z^{n \times n}$ be nonsingular with Smith
form $A$.  Any Smith massager $M \in \Z^{n \times n}$ for $A$ has the property
that $MS^{-1}$ has minimal denominator $A$.
\end{lemma}
The lemma follows directly 
from Definition~\ref{def:minden} combined with
\cite[Theorem~4]{BirmpilisLabahnStorjohann21} which shows that
the lattices $\{ v \in \Z^{1 \times n} \mid vA^{-1} \in \Z^{1 \times n} \}$
and $\{ v \in \Z^{1\times n} 
\mid vMS^{-1} \in \Z^{1 \times n} \}$ are identical.
Instead of working with the rational matrix $MS^{-1}$, we can avoid fractions
using Lemma~\ref{lem:1},
which shows that, for any $s$ that is a positive multiple of the
largest invariant factor of $A$,  the lattices  $\{ v \in \Z^{1
\times n} \mid vMS^{-1} \in \Z^{1 \times n} \}$ and $\{ v \in \Z^{1
\times n} \mid v M(sS^{-1}) \bmod s \}$ are identical. In particular,
this implies that $sA^{-1}\equiv_R M(sS^{-1})\bmod s$.

\begin{example}  \mylabel{ex:denom}
The input matrix 
$$
A =
 \left[ \begin {array}{cccc} -8&3&-1&0\\ 0&1&1&-1
\\ 4&-2&-1&-1\\ 4&-1&0&0 \end {array} \right] \in \Z^{4 \times 4}
$$
has Smith form $S=\diag(1,1,1,16 =: s)$ and 
$$
sA^{-1} =
 \left[ \begin {array}{cccc} 2&1&-1&9\\ 8&4&-4&20
\\ -8&4&-4&-12\\ 0&-8&-8&8 \end {array} \right].
$$
A reduced Smith massager for $A$ is given by 
$$
M = \left [ \begin{array}{cccc} \phantom{0} & \phantom{0} & \phantom{0} & 1 \\
 & & & 4 \\
 & & & 4 \\
 & & & 8 \end{array} \right ] \in \Z^{4 \times 4}.
$$
The Hermite form of $A$ is thus the Hermite denominator of 
the last column of $M$ divided by $s$.
This form is given in Example~\ref{ex:nontriv}.
\end{example}

\begin{remark} We say that a Smith form diagonal entry is \emph{trivial}
if it is equal to 1.  It is easy to see that the  number of nonzero
columns in a reduced Smith massager for $A$ is equal to the number of
nontrivial invariant factors of $A$.
\end{remark}

\subsection{Compact representations}

In the naive cost model, the integers 0 and 1 both require one bit
to store in their binary representation. For example, the total number of bits
required to store a nonsingular Hermite form $H \in \Z^{n \times
n}$ as a dense $n \times n$ matrix is $O(n^2 + n \log \det H)$ bits,
even if $\log \det H \ll n$.

We can save space and simplify the derivation of running time 
estimates by adopting a data structure that avoids explicitly
storing integers that are known \emph{a priori} to be zero, and by
avoiding integer multiplications where one of the operands is known
\emph{a priori} to be equal to one.
For example, 
we can avoid storing 
\emph{trivial} column of $H$ (corresponding to diagonal entry
$h_i=1$) or trivial columns of reduced Smith massagers (where $s_i = 1$).

In the proof of the following lemma, recall that we define the
bitlength of a vector to be the bitlength of the largest entry in
absolute value, as opposed to the sum of the bitlengths of the
entries.

\begin{lemma} \mylabel{remarkH}
Let $H \in \Z^{n \times n}$ be in Hermite form. Then $H$ can be
represented using $O(n \log \det H)$ bits by storing the submatrix
comprised of its nontrivial columns, together with the 
list of the indices of the nontrivial columns.
\end{lemma}
\begin{proof}
Entries in column $i$ of $H$ have magnitude bounded by the diagonal
entry $h_{i}$ of column $i$. The sum of the bitlengths of the 
\emph{nontrivial} columns of $H$ is bounded by 
\begin{eqnarray*}\mylabel{sumbit}
\sum_{\substack{i=1 \\h_i\neq 1}}^n \lg h_i &  
  \leq  & \sum_{\substack{i=1 \\h_i\neq 1}}^n (1 + \log h_i)
   \leq   \sum_{\substack{i=1\\h_i\neq 1}}^n (2\log h_i)
   =  2 \log \det H.
\end{eqnarray*}
\end{proof}
A statement similar to Lemma~\ref{remarkH} also holds for reduced Smith massagers.
\begin{lemma} \mylabel{remarkM}
Let $M \in \Z^{n \times n}$ satisfy $M =  \colmod(M, S)$ where $S
\in \Z^{n \times n}$ is a nonsingular Smith form. Then $M$ can be
represented using $O(n \log \det S)$ bits by storing only the
nontrivial columns.
\end{lemma}

\section{Computational preliminaries} \mylabel{sec:4}

In this section we define some computational tasks which will be
used later in the paper, and derive upper bounds on their complexity.
We also summarize in Subsection~\ref{ssec:prev} two results which
we need from the literature.

We consider first the computation of the remainder modulo $Y$ of
the product of two integers. Here, $b \in \Z/(Y)$ implicitly means
$b \in [0,Y)$.

\begin{lemma}  \mylabel{lem:modmult}
Let $a \in \Z$ and $b \in \Z/(Y)$.
If  $~\lg a \leq D$, then
$\Rem(ab,Y)$ can be computed in $O(D(\log Y))$ bit operations.
\end{lemma}
\begin{proof} There exists a constant $c_1$ such that
the multiplication $ab$ over $\Z$ has cost bounded by $c_1 (\lg a)(\lg b)$.
There exists a second constant $c_2$ such that $\Rem(ab,Y)$
has cost bounded by $c_2 (\lg ab/Y)(\lg Y)$.  Using
$|b| < Y$ shows that  both of these cost bounds are bounded
by $c(\lg a)(\lg Y)$ where $c=\max(c_1,c_2)$. Using $\lg a \leq D$ and
$Y>1$ we have $c(\lg a)(\lg Y) \leq c D(1+\log Y)
\leq c D(2 \log Y) \in O(D(\log Y))$.
\end{proof}

The following lemma extends Lemma~\ref{lem:modmult} by replacing the
first operand $a$ with a matrix, and the second operand $b$ with a vector.

\begin{lemma} \mylabel{lem:AB} 
Let $A \in \Z^{n \times k}$ and $b \in \Z/(Y)^{k \times 1}$.
If the sum of the bitlengths of the columns of $A$ is bounded by $D$,
then $\Rem(Ab,Y)$ can be computed in $O(nD(\log Y))$ bit operations.
\end{lemma}

\begin{proof}
Decompose $A$ into columns as $A = \left [ \begin{array}{ccc}
\vec{a}_1 & \cdots & \vec{a}_k \end{array} \right ] \in \Z^{n \times
k}$, and let $d_i$ be the bitlength of $\vec{a}_i$, $1\leq i\leq
k$.  Then, $\sum_{i}^k d_i \leq D$.  Let $b_i$ be entry $i$ 
of $b$.  Then, $$\Rem(Ab,Y) = \Rem \left (\sum_{i=1}^k \Rem(\vec{a}_i
b_i,Y),Y\right ).$$ By Lemma~\ref{lem:modmult}, there is a constant
$c$ such that computing $\Rem(\vec{a}_i b_i,Y) \in \Z/(Y)^{n \times 1}$
has cost bounded by $cnd_i (\log Y)$. Computing all $\Rem(\vec{a}_ib_i,Y)$
then has cost bounded by $\sum_{i=1}^k cnd_i(\log Y) \in O(nD(\log
Y))$.  Accumulating the sum modulo $Y$ is within this cost.
\end{proof}

The following result follows by accumulating the multiplication
cost over the rows of $A$.

\begin{corollary} \mylabel{cor:ABrows} Let $A \in \Z^{n \times k}$
and $b \in \Z/(Y)^{k \times 1}$.  If the sum of the bitlengths of
the rows of $A$ are bounded by $D$, then $\Rem(Ab,Y)$ can be computed
in $O(kD(\log Y))$ bit operations.
\end{corollary}

We now apply Lemma~\ref{lem:AB} to obtain the following result.

\begin{lemma} Given as input \mylabel{lem:applyherm}
\begin{itemize}
\item[(i)] a nonsingular Smith form $S = \diag(s_1,\ldots,s_n) \in
\Z^{n \times n}$,
\item[(ii)] a matrix  $M  \in \Z^{n \times n}$ such that $M =
\colmod(M, S)$, and
\item[(iii)] a nonsingular Hermite form $H \in \Z^{n \times n}$,
\end{itemize}
we can compute $\colmod(H M,S)$ in $O(n(\log \det S)(\log \det H))$
bit operations.
\end{lemma}

\begin{proof}
Let $\bar{M} = \colmod(HM,S)$.
If $\det S=1$, then $\bar{M}$ is the zero matrix and there is nothing
to compute. Similarly, if $\det H = 1$, then $\bar{M} = M$. Assume
therefore that $\det S,~ \det H > 1$. Note that $\bar{M} = \colmod(M
+ (H-I)M,S)$. We can thus compute $\bar{M}$ in two steps, by first computing
$B :=  \colmod((H-I)M,S)$ and then returning $\bar{M} :=  \colmod(M + B,S)$.

The second step, which adds together two matrices that are column reduced
modulo $S$, can be done in linear time, that is, in $O(n (\log\det
S))$ bit operations. It remains to bound the cost of the first step. By
Lemma~\ref{remarkH}, the sum of the bitlengths of the nonzero columns of
$H-I$ are bounded by $2 \log \det H$. Computing $B$ can be done
by premultiplying each nontrivial column of $M$ by $(H-I)$, working
modulo the corresponding diagonal entry in $S$.
By Lemma~\ref{lem:AB}, there exists a constant $c$ such that
the total cost is 
\begin{equation*}
\sum_{\substack{i=1\\s_i\neq 1}}^n 
c n(\log \det H) (\log s_j)  \in O(n(\log \det S)(\log \det H))
\end{equation*}
bit operations.
\end{proof}

The following corollary is obtained by replacing the use of
Lemma~\ref{lem:AB} with Corollary~\ref{cor:ABrows} in the proof of
Lemma~\ref{lem:applyherm}.

\begin{corollary} \label{cor:applyherm}
Given the same input as in Lemma~\ref{lem:applyherm},
we can compute $\colmod(H^T M,S)$ in\\ $O(n(\log \det S)(\log \det H))$ bit operations.
\end{corollary}

\subsection{Computing Hermite denominators and Smith massagers} 
\label{ssec:prev}

We will make use of  the following algorithms for computing the
Hermite denominator of a rational column vector and fast computation
of Smith forms and massagers. 

\begin{theorem}[\protect{\citet[Theorem~2]{PauderisStorjohann13}}] 
\mylabel{thm:hcol} 
There exists an algorithm \texttt{hcol}$(w,d)$ that takes as
input a vector $w \in \Z/(d)^{n
\times 1}$, and returns as output
the Hermite denominator $H$ of $wd^{-1}$. The cost of the
algorithm is $O(n(\log d)^2)$ bit operations. The Hermite form $H$
will satisfy $(\det H) \mid d$.
\end{theorem}

\begin{theorem}[\citet{BirmpilisLabahnStorjohann20,BirmpilisLabahnStorjohann21}] 
\mylabel{thm:sm}
There exists a Las Vegas algorithm \texttt{SmithMassager}$(A)$ that
takes as input a nonsingular $A \in \Z^{n \times n}$, and returns
as output a tuple $(M,S,p) \in (\Z^{n \times n}, \Z^{n \times n},
\Z_{>2})$ with
\begin{itemize}
\item[(i)] $S$ the Smith form of $A$,
\item[(ii)]  $M$ is a reduced Smith massager of $A$, and
\item[(iii)] $p$ is prime with $p \perp \det S$ and $\log p \in \Theta(\log n + \loglog ||A||)$.
\end{itemize}
The algorithm has cost
$O(n^3 (\log n + \log ||A||)^2(\log n)^2)$ bit operations,
using standard integer and matrix multiplication 
\end{theorem}
We remark that the prime $p$ in part~(iii) of the output specification of
Theorem~\ref{thm:sm} is needed by the subroutine developed in
Section~\ref{sec:dotprod}.

\section{Diagonal entries of the Hermite form}\label{sec:diag}\label{sec:5}

In this section, we give an algorithm for determining the diagonal
entries of the Hermite form of a nonsingular $A \in \Z^{n \times
n}$. Let the Smith form of $A$ be $S$, and suppose $M$ is a Smith
massager for $A$. The algorithm is based on Lemma~\ref{lem:key},
which states that the Hermite denominator of $MS^{-1}$ is the same
as that of $A^{-1}$.
\begin{figure}[H]
\centering
\fbox{
\begin{minipage}{.95\textwidth}
\texttt{HermiteDiagonals}$(A,M,S)$\\
{\bf Input:}
\begin{itemize}
\item[(i)] A nonsingular $A \in \Z^{n \times n}$.
\item[(ii)] The Smith form $S=\diag(s_1,\ldots,s_n)$ of $A$.
\item[(iii)] A reduced Smith massager $M$ for $A$.
\end{itemize}
{\bf Output:}
\begin{itemize}
\item[] The diagonal entries $h_1,\ldots,h_n$ of the
Hermite form of $A$.
\end{itemize}
\end{minipage}}
\caption{Problem \texttt{HermiteDiagonals} \label{fig:diag}}
\end{figure}

\begin{theorem} \mylabel{thm:diag}
Problem \texttt{HermiteDiagonals}
can be solved in $O(n (\log \det S)^2)$ bit operations.
\end{theorem}

\begin{proof}
Define $s_0 := 1$, and let $0\leq k\leq n$ be such that $s_i=1$ for
all $i\leq k$. Then, since the first $k$ columns of $MS^{-1}$ are
zero, they have minimal denominator $I_n$, and so can be ignored.
By Lemma~\ref{lem:recurse}, the following loop will compute matrices
$\hat{H}_{k+1},\hat{H}_{k+2},\ldots,\hat{H}_n$ in Hermite form such
that $\hat{H}_{n}\hat{H}_{n-1} \cdots \hat{H}_{k+1}$ is a minimal
denominator of $MS^{-1}$.\\

\noindent
\For $i=k+1$ \To $n$ \Do\\
\ind{1} $\hat{H}_i := \texttt{hcol}({\rm Column}(M,i),s_i)$\\
\ind{1} $M := \colmod(\hat{H}_i M,S)$\\
\Od\\

By Theorem~\ref{thm:hcol}, the cost of the call to \texttt{hcol}
in iteration $i$ is bounded by $cn(\log s_i)^2$ for some constant
$c$.  The total cost of all calls to \texttt{hcol} is therefore
$O(n(\log S)^2)$.  By Lemma~\ref{lem:AB}, the cost of updating $M$
during iteration $i$ is bounded by $\hat{c}n(\log \det \hat{H}_i)(\log
\det S)$ bit operations for some constant $\hat{c}$.  Since $\det
\hat{H}_i  \mid s_i$, this is bounded by $\hat{c}n(\log s_i)(\log
\det S)$.  The total cost of all updates of $M$ is then also $O(n(\log
\det S)^2)$.

While the product $\hat{H}_{n}\hat{H}_{n-1} \cdots \hat{H}_{k+1}$
is a minimal denominator of $MS^{-1}$ that is upper triangular, it
might not be in Hermite form because the off-diagonal entries might
not be reduced.  However the diagonal entries of $\hat{H}_{n}\hat{H}_{n-1}
\cdots \hat{H}_{k+1}$ will be the same as those of $H$.  Taking
advantage of our compact representation for the $\hat{H}_{\ast}$,
the total cost of computing the diagonal entries of $H$ is then
bounded by $O(n (\log \det S)^2)$.
\end{proof}

\section{Column Howell forms} \label{ssec:colhowell}

Working with matrix denominators, as discussed in Subsection
\ref{subsec:2.1}, naturally implies doing linear algebra in the
residue class ring $\R := \Z/(s)$ for a given modulus $s \in \Z_{>0}$
(c.f. Lemma~\ref{lem:1}). In this section, we investigate a type
of column echelon form for matrices in such a residue ring.

For a matrix $B \in \R^{n \times n}$, we denote by
\[\Span(B)=\{Bv\in\R^{n\times 1}\mid v\in\R^{n\times 1}\}\] the set
of all $\R$-linear combinations of the columns of $B$.  By $\Span_k(B)$
we denote the subset of $\Span(B)$ consisting of all column vectors
that have the last $k$ entries zero.

A column Howell form of $B$, first introduced by \citet{Howell}, is a
matrix $T\in\R^{n \times n}$ that is right equivalent to $B$ over $\R$
and that satisfies the \textit{Howell property}: for all $0 \leq k\leq
n$, $\Span_k(B) = \Span(T_k)$ where $T_k$ is the submatrix of $T$
comprised of those columns that have the last $k$ entries zero.

\begin{example} \mylabel{ex:howell1}
Consider the matrix
$$
B = \left [ \begin{array}{cccc} \phantom{0} & \phantom{0} & \phantom{0} & 1 \\
 & & & 4 \\
 & & & 4 \\
 & & & 8 \end{array} \right ] \in \Z/(16)^{4 \times 4}.
$$
The span of the columns of $B$ which have the last entry zero (in
this example the first three zero columns) contains only the zero
vector.
But multiplying the last column of $B$ by $2$ yields the nonzero column
$$
\left [ \begin{array}{c} 2 \\ 8 \\8 \\ \phantom{0} \end{array} \right],
$$
with last entry zero, and so $B$ does not satisfy the Howell
property. In this case the column triangularization of $B$ given by
\begin{equation} \mylabel{eq:t}
 T  =  
 \left [ \begin{array}{cccc} 
\phantom{0} & 4 & 2 & 1 \\
 & & 8 & 4 \\
 & & 8 & 4 \\
 & & & 8
\end{array} \right ]  = 
\left [ \begin{array}{cccc} \phantom{0} & \phantom{0} & \phantom{0} & 1 \\
 & & & 4 \\
 & & & 4 \\
 & & & 8 \end{array} \right ] 
\left [ \begin{array}{cccc} 
1 &  &  &  1\\
 & 1& 13 &  \\
 & & 1 & 1\\
 4 & 0& 6& 11
\end{array} \right ]  
=  B  U 
\end{equation}
with $U$ unimodular does satisfy the Howell property.
\end{example}
The Howell form is a natural generalization of the notion of the
column echelon form over a field. Variations include alternate
locations for the zero columns and/or including some additional
normalization conditions.  For our purposes, in order to simplify
the subsequent presentation, we say that a matrix $T$ is in
\emph{Howell form} if $T$ satisfies the Howell property and is upper
triangular with the diagonal entries being positive and divisors
of the modulus $s$.  The diagonal entries of the zero columns modulo
$s$ are replaced with $s$ in order to be positive.  Uniqueness of
the form can be achieved by stipulating that off-diagonal entries
are reduced modulo the diagonal entry in the same row, as per
\citet[Theorem~2]{Howell}, but we do not require this.  We will
however use the fact that the diagonal entries of a Howell form are
unique.
\begin{example} \mylabel{exhow}
Consider the matrix $B \in \Z/(16)^{4 \times n}$
from Example~\ref{ex:howell1}. A
Howell form of $B$ is obtained from the matrix $T$ in~(\ref{eq:t})
by swapping the first two columns and adding the pivot $16$ in the second
column:
$$
 \left [ \begin{array}{cccc} 
4 & 0 & 2 & 1 \\
 & 16 & 8 & 4 \\
 & & 8 & 4 \\
 & & & 8
\end{array} \right ].
$$
Another Howell form of $B$ is
$$
\left [ \begin{array}{cccc} 4 & 0 & 6 & 11 \\
 & 16 & 8 & 12 \\
 & & 8 & 12 \\
 & & & 8 \end{array} \right ] .
$$
\end{example}

\section{Solving in the dual: Hermite via Howell}
\mylabel{sec:dual}

Throughout this section, let
$$
H = \left [ \begin{array}{cccc} h_1 &h_{12} & \cdots & h_{1n} \\
 & h_2 & \cdots & h_{2n} \\
 & & \ddots & \vdots \\
 & & & h_n \end{array} \right ] \in \Z^{n \times n}
$$
be the Hermite form of $A \in \Z^{n\times
n}$. For $1\leq j\leq  n$, define
$$
\setlength{\arraycolsep}{.5\arraycolsep}\renewcommand{\arraystretch}{.75}
H_j = \left [ \begin{array}{ccccccc} 
 1   &          &        &h_{1,j}  &        &        &          \\
     & \ddots   &        & \vdots  &        &        &          \\
     &          &   1    &h_{j-1,j}&        &        &          \\
     &          &        &   h_j   &        &        &          \\ 
     &          &        &         &  1     &        &          \\ 
     &          &        &         &        &\ddots  &          \\
     &          &        &         &        &        &   1     
\end{array} \right ] \in \Z^{n \times n}
$$
to be the $n \times n$ matrix with column $j$ equal to that of $H$
and the remaining columns those of $I_n$.
Computing $H$ is thus equivalent to 
computing $H_1,\ldots,H_n$. In addition, it is useful to note that the first $j$ columns of $H_j \cdots H_1$
are those of $H$, while the last $n-j$ columns are those of $I_n$.

In this section, we establish a duality between $H$ and any Howell
form $T$ of $sA^{-1}$ over $\Z/(s)$, with $s$ a positive integer
such that $sA^{-1}$ is integral. In particular, we show that column
$j$ of $(H_{j-1} \cdots H_1)T$ is congruent modulo $s$ to
$$
-\frac{s}{h_{j}} \left [ \begin{array}{c} h_{1j} \\
 \vdots \\
 h_{j-1,j}\\
 -1 \\
\phantom{0} \\
\phantom{0} \\
\end{array} \right] \bmod s.
$$
This property points out the following algorithm for computing $H$:\\

\noindent
\For $j=1$ \To $n$ \Do\\
\ind{1} Recover $H_j$ from column $j$ of $T$\\
\ind{1} $T := \Rem(H_j T,s)$\\
\Od\\

We first show that any nonsingular upper triangular matrix over
$\Z$ corresponds to a Howell form over
$\Z/(s)$.
\begin{lemma} \mylabel{thm:notdual}
Let $T \in \Z^{n \times n}$ be nonsingular and upper triangular.
If $s \in \Z_{>0}$  is such that $sT^{-1}$ is integral, then $sT^{-1}$
satisfies the Howell property over $\Z/(s)$.
\end{lemma}

\begin{proof}
To establish the Howell property, we need to show that, for $0 \leq
k\leq n$, $\Span_k(sT^{-1})$ is equal to the span of the columns
of $sT^{-1}$ that have the last $n-k$ entries zero. To this end,
fix $k$ and decompose $T$ as $$T = \left [ \begin{array}{cc} T_1 &
\bar{T} \\
 & T_2 \end{array} \right ] $$
where $T_2 \in \Z^{k \times k}$ and the dimensions of $T_1$ and
$\bar{T}$ are implied. Then,
$$
sT^{-1}=
\left[ \begin{array}{cc}
sT_1^{-1} & -sT_1^{-1}\bar{T}T_2^{-1} \\
& sT_2^{-1}
\end{array}\right] \in \Z^{n \times n},
$$
and it will suffice to show that
\[
\Span_k(sT^{-1})
\subseteq 
\Span\left( \left[ \begin{array}{c}
sT_1^{-1} \\ \textcolor{white}{T_2}
\end{array}\right]\right).
\]
This is equivalent to saying that for any vector $v\in\Z^{n\times 1}$ such that
\[sT^{-1}v=\left[ \begin{array}{c}
\bar{v} \\ \textcolor{white}{v}
\end{array}\right]\bmod s,\]
for some $\bar{v}\in\Z^{(n-k)\times 1}$,
there exists another vector $u\in\Z^{(n-k)\times 1}$ such that
$sT_1^{-1}u=\bar{v}$.
Now, 
\begin{align}
sT^{-1}v &= \left[ \begin{array}{cc}
sT_1^{-1} & -sT_1^{-1}\bar{T}T_2^{-1} \nonumber \\
& sT_2^{-1}
\end{array}\right]
\left[ \begin{array}{c}
v_1 \\ v_2
\end{array}\right]\\ &= \label{eq:hheq1}
\left[ \begin{array}{c}
sT_1^{-1}v_1 -sT_1^{-1}\bar{T}T_2^{-1}v_2 \\ sT_2^{-1}v_2
\end{array}\right]\\ &= \label{eq:hheq2}
\left[ \begin{array}{c}
\bar{v} \\ \textcolor{white}{v}
\end{array}\right]\bmod s.
\end{align}
From the lower block of (\ref{eq:hheq1}) and (\ref{eq:hheq2}), it follows that there exist a vector $v_2'\in\Z^{k\times 1}$ such that
\[sT_2^{-1}v_2=sv_2' ~~~ \Leftrightarrow ~~~ v_2=T_2v_2'.\]
Moreover, from the upper block of (\ref{eq:hheq1}) and (\ref{eq:hheq2}), we have that
\begin{align*}
\bar{v} &= sT_1^{-1}v_1 -sT_1^{-1}\bar{T}T_2^{-1}v_2 \\
&= sT_1^{-1}v_1 - sT_1^{-1}\bar{T}v_2' \\
&= sT_1^{-1}(v_1-\bar{T}v_2'),
\end{align*}
which proves the claim.
\end{proof}

\begin{corollary}   \mylabel{corhowherm}
If $H$ is the Hermite form of $A$,  then $sH^{-1}$ is a Howell form of
$sA^{-1}$ over $\Z/(s)$.
\end{corollary}
\begin{proof}
The result follows since $sH^{-1} \equiv_R sA^{-1}$, $sH^{-1}$ is
upper triangular, the diagonal entries of $sH^{-1}$ are positive
divisors $s/h_1,s/h_2,\ldots,s/h_n$ of $s$ and, from
Lemma~\ref{thm:notdual}, $sH^{-1}$ satisfies the Howell property.
\end{proof}

\begin{corollary}  \mylabel{lem:diags}
The diagonal entries of any any Howell form of $sA^{-1}$
over $\Z/(s)$ are equal to $s/h_1,\ldots,s/h_n$.
\end{corollary}
\begin{proof}
This follows from Corollary~\ref{corhowherm} and the fact that 
the diagonal entries of a Howell form of $sA^{-1}$ are unique.
\end{proof}

\begin{lemma}  \mylabel{lem:denom}
Let $T$ be a Howell form of $sA^{-1}$ over $\Z/(s)$.
Then, $H_j \cdots H_1$ is a denominator of the first $j$
columns of $(1/s)T$, for $1\leq j\leq n$.
\end{lemma}

\begin{proof}
Since $T$ is right equivalent to $sA^{-1}$ over $\Z/(s)$, and $H$ is
a denominator of $A^{-1}$, we have that $H$ is a denominator
of $(1/s)T$.
The claim in the lemma now follows from the fact
that $T$ is upper
triangular.  In particular, premultiplying an upper triangular
matrix by $H_k$ for $k > j$ does not change the first $j$ columns.
Let $T_{1\ldots j}$ denote the submatrix of $T$ comprised of the first $j$ columns.  Then ,
\begin{eqnarray}
\label{eq:HT1}
 HT_{1..j} & = & (H_n\cdots H_{j+1})(H_{j}\cdots
H_1)T_{1..j}\\
\label{eq:HT2}  & = & (H_{j}\cdots H_1)T_{1..j}.
\end{eqnarray} 
Since the left hand side of~(\ref{eq:HT1}) is zero modulo $s$,
so is the right hand side of~(\ref{eq:HT2}).
\end{proof}

\begin{theorem} \mylabel{lemidea} 
Let $T$ be a Howell form of $sA^{-1}$ over $\Z/(s)$.
Then, $H_j \cdots H_1$ is the minimal Hermite denominator
of the first $j$ columns of $(1/s)T$, for $1\leq j\leq  n$.
\end{theorem}

\begin{proof}
Recall that we let $T_{1\ldots j}$ denote the first $j$ columns of
$T$. We will use induction. For $j=1$, the claim of the theorem
follows from Corollary~\ref{lem:diags} and Lemma~\ref{lem:denom},
since the first diagonal entry of $T$ is $s/h_1$ and $H_1$ is a
denominator of $(1/s)T_1$.

Now,
assume that the claim is true for $j-1$, for some $j > 1$, that is,
assume that $H_{j-1} \cdots H_{1}$ is the minimal Hermite denominator of
$(1/s)T_{1\ldots j-1}$. Then, let $v$ be
column $j$ of $(H_{j-1} \cdots H_1) T$.
We first show that $v$ has the shape
\begin{equation} \mylabel{eq:coljT}
v = \left [ \begin{array}{c} \ast \\
 \vdots \\
 \ast \\
 s/h_j \\
\phantom{0} \\
\phantom{0} \\
\end{array} \right ] \in \Z/(s)^{n \times 1}.
\end{equation}
To see this, note that premultiplying $T$ by $H_{j-1} \cdots H_1$
only affects the first $j-1$ rows, so the last $n-j+1$ entries of
$v$ are the same as those of column $j$ of $T$. By
Lemma~\ref{lem:diags}, entry $j$ of $v$ is equal to $s/h_j$.

Next, by Lemma~\ref{lem:recurse}, a minimal
denominator of $(1/s)T_{1\ldots j}$ is given by $\bar{H}_jH_{j-1}
\cdots H_1$, where $\bar{H}_j$ is the minimal Hermite denominator of
$(1/s)v$. By Lemma~\ref{lem:denom},
$H_j H_{j-1} \cdots H_1$ is a denominator
of $(1/s)T_{1\ldots j}$, so we must have that $\det \bar{H}_j$ is
a divisor of $\det H_j = h_j$.
But since entry $j$ of $(1/s)v$ 
is $1/h_j$, the diagonal entry $j$
 of $\bar{H}_j$ must equal $h_j$, which means that the remaining columns of $\bar{H}_j$
have diagonal entry $1$. Because of the shape of $\bar{H}_j$,
and the fact that it is in Hermite form,
we have that $\bar{H}_j H_{j-1} \cdots H_1$ is also in Hermite form.
The uniqueness of the Hermite form then implies that $\bar{H}_j= H_j$.  
\end{proof}

\begin{corollary} \mylabel{cor:getH}
For $1\leq j \leq n$, column $j$ of $(H_{j-1}\cdots H_1)T$ is equal to 
\begin{equation} \label{coljj}
-\frac{s}{h_{j}} \left [ \begin{array}{c} h_{1j} \\
 \vdots \\
 h_{j-1,j}\\
 -1 \\
\phantom{0} \\
\phantom{0} \\
\end{array} \right] \bmod s.
\end{equation}
\end{corollary}

\begin{proof}
Column $j$ of $(H_{j-1} \cdots H_1)T$ is the vector $v \in \Z/(s)^{n
\times 1}$ in (\ref{eq:coljT}), from
the proof of Theorem~\ref{lemidea}, where it was established
that entry $j$ of $v$ is $s/h_j$, the
last $n-j$ entries of $v$ are zero, and $H_j$ is the minimal Hermite
denominator of $(1/s)v$.
The only such vector $v$ is the one shown in~(\ref{coljj}),
namely, column $j$ of $sH_j^{-1}$.
\end{proof}

The following example illustrates the approach of
Corollary~\ref{cor:getH} for computing the Hermite form over $\Z$
by first computing a Howell form in the space $\Z/(s)$.
\begin{example} \mylabel{exampd1}
The input matrix 
$$
A = 
 \left[ \begin {array}{cccc} -8&3&-1&0\\ 0&1&1&-1
\\ 4&-2&-1&-1\\ 4&-1&0&0 \end {array} \right] \in \Z^{4 \times 4}
$$
has Smith form $S=\diag(1,1,1,16 =: s)$ and
$$
sA^{-1} = 
 \left[ \begin {array}{cccc} 2&1&-1&9\\ 8&4&-4&20
\\ -8&4&-4&-12\\ 0&-8&-8&8 \end {array} \right].
$$
We now work over $\Z/(s)$.  A Howell form of $sA^{-1}$
over $\Z/(s)$ is given by
$$T =  \left [ \begin{array}{cccc} 4 & 0 & 6 & 11 \\
 & 16 & 8 & 12 \\
 & & 8 & 12 \\
 & & & 8 \end{array} \right ] 
  =  \left [ \begin{array}{cccc} \frac{s}{4} & 0 & 6 & 11 \\
 & \frac{s}{1} & 8 & 12 \\
 & & \frac{s}{2}  & 12 \\
 & & & \frac{s}{2} \end{array} \right ]   = 
  \left [ \begin{array}{cccc} \frac{s}{h_{1}}  & 0 & 6 & 11 \\
 & \frac{s}{h_{2}}  & 8 & 12 \\
 & & \frac{s}{h_{3}}  & 12 \\
 & & & \frac{s}{h_{4}} \end{array} \right ].$$
The diagonal elements of $H$ are thus $h_1,h_2,h_3,h_4 = 4,1,2,2$.
Using Corollary~\ref{cor:getH} gives the following:
\begin{eqnarray*}
j=1: & & 
H_1 = \left [ \begin{array}{cccc} 4 & \phantom{0} & \phantom{0} & \phantom{0} \\
 & 1 & & \\ 
& & 1 & \\
& & &  1 \end{array} \right]\mbox{~~~and~~~}
H_1 T = 
\left [ \begin{array}{cccc}   & 0 & 8 & 12 \\
 & 16 & 8 & 12 \\
 & & 8 & 12 \\
 & & & 8 \end{array} \right ]\\
j=2: & & H_2 = \left [ \begin{array}{cccc}
1 & 0 & \phantom{0} & \phantom{0} \\
 & 1 & & \\
& & 1 & \\
 & & & 1 \end{array} \right ] \mbox{~~and~~}
 H_2 H_1 T = H_1 T\\
j=3: & & H_3  = \left [ \begin{array}{cccc}
1 & \phantom{0} & 1 & \phantom{0} \\
 & 1 &1 & \\
& & 2 & \\
 & & & 1 \end{array} \right ] \mbox{~~and~~}
H_3 H_2 H_1  T = 
\left [
\begin{array}{cccc} \phantom{0} & \phantom{0} & \phantom{0} & 8 \\
 & & & 8 \\
 & & & 8 \\
 & & & 8 \end{array} \right ]\\
j=4: & & H_4  = \left [ \begin{array}{cccc}
1 & \phantom{0} &   & 1 \\
 &  1  &  &1 \\
& & 1 & 1 \\
 & & & 2 \end{array} \right ] \mbox{~~and~~}
H_4 H_3 H_2 H_1  T  = 0_{4 \times 4}
\end{eqnarray*}
The Hermite denominator of $(1/s)T$ is thus 
$$ H_4 H_3 H_2 H_1 = 
\left [ \begin{array}{cccc}
4 & 0 & 1 & 1 \\
 & 1 & 1 & 1\\
 & & 2 & 1\\
& & & 2 
\end{array} \right ].$$
\end{example}

\section{Computing a Hermite form from a Howell form}\label{sec:8}

For $A \in \Z^{n \times n}$ nonsingular with Smith form $S$, let
$s= s_n$ be the largest invariant factor of $S$, and let ${S}^*
= s S^{-1}$ and $A^* = s A^{-1}$.
A problem with using the approach of Example~\ref{exampd1} to compute
$H$, is that the size of $A^*$ and its Howell form $T$ over
$\Z/(s)$ can be $\Omega(n^2 \log s)$ bits.

In this section, we show how we can avoid computing the Howell form
$T$ explicitly, and instead work with matrices $M,U\in\Znn$ such
that \[T=MS^*U\bmod s.\] 

We start first with $M{S}^*$, where $M$ is a reduced
Smith massager for $A$, which we know is
right equivalent to $A^*$
over $\Z/(s)$ but has total size only $O(n \log \det S)$ bits, as per
Lemma~\ref{remarkM}. 
We then compute a transformation matrix $U$ such that $T = MS^* U\bmod s$.
\begin{lemma} Let $U \in \Z/(s)^{n \times n}$ be 
such that $T=M{S}^*U$ is a Howell form of $M{S}^*$
over $\Z/(s)$.  Then $T=M{S}^* \rowmod(U,S)$.
\end{lemma}
Thus, we may assume without loss of generality that $U= \rowmod(U,S)$.
So, while the overall size of $T$ itself can be large, the
transformation matrix $U$ to generate $T$ can be assumed to be small, that is, just
like $M$, it can be represented using $O(n \log \det S)$ bits.
The following example illustrates how a Howell form $T$
can be represented implicitly as the product $M {S}^* U$.
\begin{example} \mylabel{ex:toolarge} 
The input matrix
$$
A = \left [ \begin{array}{ccccc} 2 & -1 & & & \\
 & 2 & -1 & & \\
 & & 2 & \ddots &     \\
 & & & \ddots & -1 \\
 & & & &  2 
\end{array} \right ] \in \Z^{n \times n}
$$
has Smith form $\diag(1,\ldots,1,2^n =: s)$ 
and
$$
A^* := sA^{-1} = \left [ \begin{array}{cccccc}
2^{n-1} & 2^{n-2} & 2^{n-3} & \cdots & 1 \\
 & 2^{n-1} & 2^{n-2} & \cdots & 2^1 \\
&  & 2^{n-1} &\cdots & 2^2 \\
 &  & & \ddots & \vdots \\
 &  & & & 2^{n-1}
\end{array} \right ].
$$
By Lemma~\ref{thm:notdual},
$A^*$ is in Howell form over $\Z/(s)$. 
The sum of the bitlengths of entries
in $A^*$ is clearly $\Theta(n^3)$. 

However, a reduced Smith massager
for $A$ is given by the $n \times n$ matrix
$$
M = \left [ \begin{array}{ccccc}
\phantom{0}& \phantom{0} & \phantom{0}& \phantom{0} & 1 \\
\phantom{0}& \phantom{0} & \phantom{0}& \phantom{0} & 2^1 \\
\phantom{0}& \phantom{0} & \phantom{0}& \phantom{0} & 2^2 \\
\phantom{0}& \phantom{0} & \phantom{0}& \phantom{0} & \vdots \\
\phantom{0}& \phantom{0} & \phantom{0}& \phantom{0} & 2^{n-1} 
\end{array} \right ].
$$
Let ${S}^*:= s S^{-1}$. A matrix $U$ such that
$A^* = M{S}^*U$ is given by
$$
U = \left [ \begin{array}{cccccc} \phantom{1} & & & & \\
 & \phantom{1} & & & \\
 & & \phantom{1} & & \\
 & & & \phantom{\ddots} & \\
 & & &        & \phantom{1} \\
2^{n-1} & 2^{n-2} & 2^{n-3} & \cdots  & 2^1 & 1 
\end{array} \right ].
$$
The sum of the bitlengths of all entries in $M$ and $U$ is only $O(n^2)$.
Restricting $M$ and $U$ to their nonzero columns and rows, respectively,
and restricting $S^*$ to its only nonzero entry, gives
$$
A^* = M S^* U = \left [ \begin{array}{c} 1 \\
2^1 \\
2^2 \\
\vdots \\
2^{n-1} \end{array} \right ] 
\left [ \begin{array}{c} 1 \end{array} \right ]
\left [ \begin{array}{cccccc} 2^{n-1} & 2^{n-2} & 2^{n-3} &
\cdots & 2^1 & 1 \end{array} \right ] ~\mod~s.
$$
\end{example}

Instead of working with an explicit Howell form $T$ of $A^*$,
we work with the right hand side of the equation $T = M{S}^*U$.
At iteration $j$, we then compute column $j$ of $-(h_j/s)\Rem(M{S}^* U,s)
\in \Z/(h_j)^{n \times 1}$ which gives the off-diagonal entries in
column $j$ of $H$.  Finally, to update $T$ at iteration $j$ we simply
update $M := \colmod(H_j M,S)$.

\begin{example} \mylabel{exampd}
The input matrix 
$$
A = 
 \left[ \begin {array}{cccc} -8&3&-1&0\\ 0&1&1&-1
\\ 4&-2&-1&-1\\ 4&-1&0&0 \end {array} \right] \in \Z^{4 \times 4}
$$
has Smith form $S=\diag(1,1,1,16 =: s)$.  Since in this example
$A$ has only one nontrivial
invariant factor, a reduced
Smith massager $M$ for $A$ and transformation matrix $U$ such
that $M{S}^*U$ is in Howell form will have one nonzero column and row respectively.
Restrict $M$ to its last column, $U$ to its last row, set 
$S= \diag(16)$ and ${S}^* = \diag(1)$.
Then
$$
T = M {S}^* U = \left[ \begin {array}{c} 1\\ \noalign{\medskip}4\\ \noalign{\medskip}4
\\ \noalign{\medskip}8\end {array} \right] \left [ \begin{array}{c} 1 \end{array} \right ]
\left[ \begin {array}{cccc} 4&0&6&11\end {array} \right].
$$
Suppose we have precomputed the diagonal entries $h_1,h_2,h_3,h_4 =
4,1,2,2$ of the Hermite denominator of $MS^{-1}$.
Applying the approach of Corollary~\ref{cor:getH} gives the following:
\begin{eqnarray*}
j=1: & - \frac{h_1}{s}{\rm Column}(M{S}^* U,1) = \left [ \begin{array}{c} -1 \\
 \phantom{0} \\
\phantom{0} \\
\phantom{0}  \end{array} \right ] &
M := \colmod(H_1 M,S) =  \left [ \begin{array}{c} 4 \\ 
4 \\
4 \\
8 \end{array} \right ]\\
j=2: & 
-\frac{h_2}{s} {\rm Column}(M {S}^* U,2) =  \left [ \begin{array}{c} 0 \\ -1 \\
\phantom{0} \\
\phantom{0} \end{array} \right ] 
&
M := \colmod(H_2M,S) = 
\left [ \begin{array}{c}   4 \\
 4 \\
 4 \\
 8\end{array} \right ] \\
j=3: & 
-\frac{h_3}{s} {\rm Column}(M {S}^* U,3) =  \left [ \begin{array}{c} 1 \\ 1 \\
-1  \\
\phantom{0} \end{array} \right ]  &
M := \colmod(H_2M,S) = 
\left [ \begin{array}{c}   8 \\
 8 \\
 8 \\
 8\end{array} \right ] \\ 
j=4: & 
-\frac{h_4}{s} {\rm Column}(M {S}^* U,4) =  \left [ \begin{array}{c} 1 \\ 1 \\
1 \\
-1  
\end{array} \right ]  &
M := \colmod(H_3 M,S) = \left [\begin{array}{c} 0 \\
0 \\
0 \\
0 \end{array} \right ] .
\end{eqnarray*}
The Hermite basis of $A$ is thus given by 
$$ H_4 H_3 H_2 H_1 = 
\left [ \begin{array}{cccc}
4 & 0 & 1 & 1 \\
 & 1 & 1 & 1\\
 & & 2 & 1\\
& & & 2 
\end{array} \right ].$$
\end{example}

\begin{figure}[H]
\centering
\fbox{
\begin{minipage}{.95\textwidth}
\texttt{HermiteViaHowell}$(A, M, S, U, [h_1,\ldots,h_n],p)$\\
{\bf Input:}
\begin{itemize}
\item[(i)] A nonsingular $A \in \Z^{n \times n}$.
\item[(ii)] The Smith form  $S = \diag(s_1,\ldots,s_n)$ of $A$.
\item[] \- Let $s := s_n$ and ${S}^* := sS^{-1}$.
\item[(iii)] A reduced Smith massager $M$ for $A$.
\item[(iv)] A $U \in \Z^{n \times n}$ such that $\Rem(M{S}^*U,s)$ is in
Howell form over $\Z/(s)$ and $U = \rowmod(U,S)$.
\item[(v)] The diagonal entries $h_1,\ldots,h_n$ of the Hermite form of $A$.
\item[(vi)] A prime $p$ that satisfies $p \perp s$ and $\log p \in O(\loglog S)$.
\end{itemize}
{\bf Output:}
\begin{itemize}
\item[] The Hermite form $H$ of $A$.
\end{itemize}
\end{minipage}}
\caption{Problem \texttt{HermiteViaHowell} \label{fig:hvh}}
\end{figure}

\begin{theorem} \mylabel{thm:hvh}
Problem \texttt{HermiteViaHowell} can be solved in 
$$O(n(\log \det S)^2 + n^2(\log \det S) (\loglog \det S))$$
bit operations.
\end{theorem}

\begin{proof}
By Corollary~\ref{cor:getH}, we can compute $H_1,\ldots,H_n$
iteratively  as follows:\\

\noindent
\For $j=1$ \To $n$ \Do
\begin{enumerate}
\item[] \# If $h_j=1$ then set $H_j := I_n$ and go to next loop iteration.
\item[1.] \# Let $u \in \Z/(s)^{1 \times n}$ be column $j$ of $U$.\\
$v := -(h_j/s)\Rem(M{S}^*u,s)$
\item[2.] \# Construct $H_j$ from $v$ and $h_j$.\\
$M := \colmod(H_j M,s)$
\end{enumerate}
\Od\\

For the construction of $H_j$ in Step~2, the off-diagonal entries in column $j$
are given by the first $j-1$ entries of $v$, and $h_j$ is given as input.
The proof of Theorem~\ref{thm:diag} shows that the total cost of the updates
to $M$ in Step~2 is bounded by $O(n (\log \det S)^2)$ bit operations.

In Section~\ref{sec:dotprod} we develop an algorithm \texttt{ScaledMatVecProd}
that will compute $v$ in Step~1 during iteration $j$ with the call
$$
v := \texttt{ScaledMatVecProd}(M,S,u,h_j,p).
$$
The \texttt{ScaledMatVecProd}
algorithm exploits the properties $M = \colmod(M,S)$, $u = \rowmod(u,s)$,
the product $M{S}^*u$ is only required modulo $s$, and that
$\Rem(M{S}^*u,s)$ has a known factor $s/h_j$.
We show later in Theorem~\ref{thm:smvp},  that \texttt{ScaledMatVecProd}
has cost
\begin{equation} \mylabel{eq:cdkdk}
O(n(\log \det S)(\log h_j + \loglog \det S) + (\log \det S)^2)
\end{equation}
bit operations. Since $\prod_{i=1}^n h_j = \det S$, 
the sum of~(\ref{eq:cdkdk})
over over all  $h_j$ with $j>1$ is bounded by the cost stated in
the theorem.
\end{proof}

\section{Computing the multiplier for a Howell form}
\mylabel{sec:comphowtrans}

In this section, we work over the residue class ring $\R=\Z/(s)$
for a given modulus $s \in \Z_{>0}$.
\citet{Howell} gives an algorithm to compute a Howell form $T$ of
a $B\in \R^{n \times n}$.  Here we adapt Howell's approach
to our context. In particular, instead of $T$, we focus on the
invertible  transformation matrix $U \in \R^{n \times n}$ such that
$T=BU$. Furthermore, we know positive divisors $h_1,\ldots,h_n$ of $s$
such that the diagonal entries of the Howell form are $t_1,\ldots,t_n$,
with $t_i=s/h_i$.

Howell's algorithm begins by augmenting the input matrix with $n$
initial zero columns: to this end, let $\bar{B} := \left[\begin{array}{c|c}
0_{n \times n} & B \end{array}\right]\in\R^{n\times 2n}$.  Our goal now is
to find a  matrix $\bar{U} \in \R^{2n \times 2n}$ such
that $\bar{B}\bar{U}$ is a Howell form of $\bar{B}$, as defined
before Example~\ref{exhow}.
Once $\bar{U}$ has been found, we can take $U$ to be the trailing principal
$n \times n$ submatrix of $\bar{U}$.  Then $BU$ will be a Howell
form of $B$.

Howell's algorithm proceeds in $n$ iterations, for $i=0,1,\ldots,n-1$.
We initialize $\bar{U} = I_{2n}$. At the start of iteration $i=0$
we thus have $\bar{B}\bar{U}=\bar{B}$.  By the time we reach
the start of iteration
$i$, the matrix $\bar{U}$ has been updated so that
\begin{equation} \mylabel{eq:parthowell}
\bar{B}\bar{U} = 
\left[\begin{array}{cc|cccc|ccc}
&& \ast & \cdots & \ast & \ast & \ast & \cdots & \ast \\
&& \vdots & & \vdots & \vdots & \vdots & & \vdots \\
&& t_{n-i}a_1 & \cdots & t_{n-i}a_{n-1} & t_{n-i}a_n & \ast & \cdots & \ast \\\hline
&&  &  &  &  & t_{n-i+1} & \cdots & \ast \\
&&  &  &  &  & & \ddots & \vdots \\
&&&&&&&& t_n
\end{array}\right]. 
\end{equation}
Note that we do not compute the complete partial triangularization
$\bar{B}\bar{U}$ in~(\ref{eq:parthowell}). We will see that we only
need the elements $a_1,\ldots,a_n \in \Z/(h_{n-i})$ shown
in~(\ref{eq:parthowell}).  Since we are working modulo $s$ and
$h_{n-i} = s/t_{n-i}$, the integers $a_1,\ldots,a_n$ can be considered
to be elements of $\Z/(h_{n-i})$.
Iteration $i$ now applies the following two-part unimodular column
transformation.
\citet{Howell} points out that there exist integers
$c_1,\ldots,c_{n-1},c_n\in\Z/(h_{n-i})$, with $c_n$ relatively
prime to $s$, satisfying 
$$c_1 a_1 + \cdots +c_{n-1}a_{n-1} + c_n a_n = 1 \bmod h_{n-i}.$$
Postmultipying the matrix on the right of~(\ref{eq:parthowell})
by the matrix
\begin{equation}
C_i = 
\left[\begin{array}{c|cccc|c}
I_{n-i} &&&&& \\\hline
& 1 & & & c_1 &\\
&& \ddots & & \vdots & \\
&&& 1 & c_{n-1} & \\
&&&& c_n & \\\hline
&&&&& I_i
\end{array}\right] 
\label{eq:howellC}
\end{equation}
gives
\begin{equation} \mylabel{eq:dok}
\left[\begin{array}{cc|cccc|ccc}
&& \ast & \cdots & \ast & \ast & \ast & \cdots & \ast \\
&& \vdots & & \vdots & \vdots & \vdots & & \vdots \\
&& t_{n-i}a_1 & \cdots & t_{n-i}a_{n-1} & t_{n-i} & \ast & \cdots & \ast \\\hline
&&  &  &  &  & t_{n-i+1} & \cdots & \ast \\
&&  &  &  &  & & \ddots & \vdots \\
&&&&&&&& t_n
\end{array}\right].
\end{equation}
Thus, we can use $t_{n-i}$ to zero out the nonzero entries to the
left of $t_{n-i}$ This step also fills in a new column which is to be
used in the subsequent iterations. If we were to postmultiply the matrix
in~(\ref{eq:dok}) by
\begin{equation}
W_i := 
\left[\begin{array}{cc|cccc|c}
I_{n-i-1} &&&&&& \\
& 1 && & & &\\\hline
&& 1 & & & &\\
&&& \ddots & & & \\
&&&& 1 & & \\
& h_{n-i}&-a_1& \cdots& -a_{n-1}& 1 & \\\hline
&&&&&& I_i
\end{array}\right] 
\label{eq:howellW}
\end{equation}
then we would obtain
$$
\left[\begin{array}{cc|cccc|ccc}
&\ast & \ast & \cdots & \ast & \ast & \ast & \cdots & \ast \\
&\vdots & \vdots & & \vdots & \vdots & \vdots & & \vdots \\
&\ast & \ast & \cdots & \ast & \ast & \ast & \cdots & \ast \\
&& & & & t_{n-i} & \ast & \cdots & \ast \\\hline
&&  &  &  &  & t_{n-i+1} & \cdots & \ast \\
&&  &  &  &  & & \ddots & \vdots \\
&&&&&&&& t_n
\end{array}\right].
$$
Note that we can do these triangularizations implicitly as  we only really need the $a_{\ast}$. 
The final computational part, at step $i$, is to update $\bar{U} := \bar{U}C_i W_i$.

At iteration $i$, the Howell transform algorithm  thus has three
steps:

\begin{enumerate}
\item Compute the entries $\left [ \begin{array}{ccc} a_1 & \cdots & a_n 
\end{array} \right ] \in \Z/(h_{n-i})^{1 \times n}$ of~(\ref{eq:parthowell}).
\item Compute the matrices $C_i$ and $W_i$.
\item Update $\bar{U} := \bar{U} C_i W_i$.
\end{enumerate}
Note that if $h_{n-i}=1$ then iteration $i$ can be skipped since
$C_i$ and $W_i$ will be the identity matrices.
For Step (2) we can appeal to the following result.
\begin{lemma} \mylabel{smlem}
Given integers $h_{n-i} \in\Z_{>0}$ and $a_1,\ldots,a_n\in\Z/(h_{n-i})$,
we can compute the off-diagonal nonzero entries of matrices $C_i,
W_i \in\Z^{2n\times 2n}$ as seen in (\ref{eq:howellC}) and
(\ref{eq:howellW}), respectively, in time $O(n(\log h_{n-i})^2)$.
\end{lemma}
\begin{proof}
\citet[Lemma~2]{StorjohannMulders98:1} show that
the $c_1,\ldots,c_{n-1}$ can be computed in the allotted time, with
$c_n$ just an extra gcd operation over $\Z/(h_{n-i})$.
Computing the entries of $W_i$ just involves negating the $a_i$.
\end{proof}

For the analysis of Steps (1) and (3), we consider the special case
of an input matrix $B=M {S}^*$ as specified in Figure~\ref{fig:sht}.

\begin{figure}[H]
\centering
\fbox{
\begin{minipage}{.95\textwidth}
\texttt{SpecialHowellTransform}$(A,M,S,[h_1,\ldots,h_n],p)$\\
{\bf Input:}
\begin{itemize}
\item[(i)] A nonsingular $A \in \Z^{n \times n}$.
\item[(i)] The Smith form  $S =
 \diag(s_1,\ldots,s_n)$ of $A$.
\item[] \- Let $s :=s_n$ and ${S}^* := s S^{-1}$.
\item[(ii)] A reduced Smith massager $M$ for $A$.
\item[(iii)] The diagonal entries $h_1,\ldots,h_n$ of the Hermite
form of $A$.
\item[(iv)] A prime $p$ that satisfies 
$p \perp s$ and $\log p \in \Theta(\loglog s)$.
\end{itemize}
{\bf Output:}
\begin{itemize}
\item[] A matrix $U = \rowmod(U,S)  \in \Z^{n \times n}$ such that
$M {S}^*U$ is a Howell form of $M {S}^*$ over $\Z/(s)$.
\end{itemize}
\end{minipage}}
\caption{Problem \texttt{SpecialHowellTransform} \label{fig:sht}}
\end{figure}

\begin{theorem} \mylabel{thm:sht}
Problem \texttt{SpecialHowellTransform} can be solved in 
in $$O(n(\log \det S)^2 + n^2(\log \det S) (\loglog \det S))$$
bit operations.
\end{theorem}

\begin{proof}
We adapt Howell's algorithm described at the start of this section
to compute an $n \times 2n$ matrix $\bar{U}$ such $M S^* \bar{U} =
\left [ \begin{array}{cc} 0_{n \times n} & T \end{array} \right ]$,
with $T$ a Howell form of $MS^*$ over $\Z/(s)$. Our output $U$ is
thus the submatrix comprised of the last $n$ columns of $\bar{U}$.
Because of the presence of the scaling matrix $S^*$, we can keep
the rows of $\bar{U}$ reduced modulo the corresponding diagonal
entries in $S$. In other words, we maintain $\bar{U} = \rowmod(\bar{U},S)$
throughout the algorithm.

Initialize $\bar{U}= \left [ \begin{array}{c|c} 0_{n \times n} &
I_n \end{array} \right ]$.  We perform $n$ iterations for
$i=0,1,\ldots,n-1$.  At the start of iteration $i$ the matrix $MS^*
\bar{U}$ has exactly the shape shown in~(\ref{eq:parthowell}). 
Like before, iteration $i$ consists of three steps:
\begin{enumerate}
\item Compute the entries $\left [ \begin{array}{ccc} a_1 & \cdots & a_n 
\end{array} \right ] \in \Z/(h_{n-i})^{1 \times n}$ of~(\ref{eq:parthowell}).
\item Compute the matrices $C_i$ and $W_i$.
\item Update $\bar{U} := \rowmod(\bar{U} C_i,S)$ and then $\bar{U} := \rowmod(\bar{U} W_i,S)$.
\end{enumerate}

At iteration $i$, the computation of Step~1 aligns with the
specification of the \texttt{ScaledMatVecProd} subroutine that is
later developed in Section~\ref{sec:dotprod}. In particular, 
the output of
$$\mbox{\texttt{ScaledMatVecProd}}(M',S,u',h_{n-i},p),$$ where
\begin{itemize}
\item $M'$ is the transpose of the submatrix of $\bar{U}$ containing columns from $(n-i+1)$ to $(2n-i)$, and
\item $u'$ is the transpose of row $n-i$ of $M$,
\end{itemize}
contains exactly the $a_i$'s we want. The cost of this call to \texttt{ScaledMatVecProd}
is 
\begin{equation} \mylabel{eq:cs1}
O(n(\log \det S)(\log h_{n-j} + \loglog \det S) + (\log \det S)^2)
\end{equation}
bit operations (Theorem~\ref{thm:smvp}).

By Lemma~\ref{smlem}, the cost of Step~2 is
\begin{equation} \mylabel{eq:cs2}
O(n(\log h_{n-i})^2)
\end{equation}
bit operations.

Finally, the two multiplications in Step~3, namely,
$\rowmod(\bar{U}C_i,S)$ and $\rowmod((\bar{U}C_i)W_i,S)$,
are covered by Corollary~\ref{cor:applyherm} and Lemma~\ref{lem:applyherm},
respectively, and have cost bounded by 
\begin{equation} \mylabel{eq:cs3}
O(n(\log \det S)(\log h_{n-i}))
\end{equation} 
bit operations.

Summing~(\ref{eq:cs1}), (\ref{eq:cs2}) and (\ref{eq:cs3}) over 
all iterations $i$ with $h_{n-i}>1$ gives the cost bound stated in the theorem.
\end{proof}

\begin{example}
Let
$$
A = \left[ \begin {array}{cccc} -13&10&-20&27\\ {}27&30&
15&30\\ {}0&15&15&6\\ {}-21&0&-15&9
\end {array} \right],
$$
with Smith form
$S = \diag(1,3,15,105)$ and reduced Smith massager
$$
M = \left[ \begin {array}{cccc} 0&2&0&55\\ {}0&0&7&32
\\ {}0&2&2&41\\ {}0&2&10&10
\end {array} \right]
$$
be given. We are also given the diagonal entries $h_1,h_2,h_3,h_4
= 1,15,15,21$ of the Hermite form of $A$.  Let ${S}^* = sS^{-1}$
with $s=105$.  We illustrate the method used in the proof of
Theorem~\ref{thm:sht} to compute a matrix $U$ such that $T = MS^*
U$ is in Howell form over $\Z/(s)$.  Note that we know that $T$
will have diagonal entries $t_1,t_2,t_3,t_4 = 105, 7, 7, 5$, that
is, $t_i = s/h_i$.

Initialize $\bar{U} = \left [ \begin{array}{cc} 0_{n \times  n} &
I_n \end{array} \right ]$. 
At the start of iteration $j=0$ we have
$$
MS^* \bar{U} = \left [ \begin{array}{cccccccc} 
\phantom{0} & \phantom{0} & \phantom{0} & \phantom{0} & \ast &  \ast & \ast & \ast \\
& & & & \ast & \ast & \ast & \ast \\
& & & & \ast & \ast & \ast & \ast \\
& & & &  0 t_4 & 14t_4  & 14 t_4& 2 t_4
\end{array} \right ],
$$
with $[a_1,a_2,a_3,a_4] = [ 0, 14,14, 2]$.
Working over $\Z/(21)$, we solve the system 
$$
c_1 0 + c_2 14 + c_3 14 + c_4 2 = 1 \bmod 21,
$$
to obtain $c=[c_1,c_2,c_3,c_4] = [0,2,2,4]$. Thus, $C_0$ and $W_0$ are
$$
\left[ \begin {array}{cccccccc} 1&&&&&&&\\ 
&1&&&&&&\\ &&1&&&&&
\\ &&&1&&&&\\ &&&&1&&
&0\\ &&&&&1&&2\\ &&&&&
&1&2\\ &&&&&&&4\end {array} \right]
 \mbox{ and }~\left[ \begin {array}{cccccccc} 1&&&&&&&\\ 
&1&&&&&&\\ &&1&&&&&
\\ &&&1&&&&\\ &&&&1&&&\\ &&&&&1&&\\ &&&&&
&1&\\ &&&21&0&-14&-14&1\end {array} \right] ,
$$
respectively.
After updating $\bar{U} = \rowmod(\bar{U}C_0 W_0,S)$ we have
$$
\bar{U} = 
\left[ \begin {array}{cccccccc} \phantom{0}&\phantom{0}&\phantom{0}
 &0&0&0&0&0\\ 
&&&0&0&0&2&2\\ &&&12&0&2&3&2
\\ &&&84&0&49&49&4\end {array} \right].
$$

Now we move on to iteration $j=1$. We have
$$
MS^* \bar{U} = 
\left[ \begin {array}{cccccccc} \phantom{0} &\phantom{0} &\phantom{0}
 &\ast&\ast&\ast&\ast&\ast \\&&&\ast&\ast&\ast&\ast&\ast\\ 
&&&6t_3&0t_3 &6t_3&13t_3&\ast\\ &&&&&&&t_4
\end {array} \right]$$
with $[a_1, a_2, a_3, a_4] = [6, 0, 6, 13]$. Working over $\Z/(15)$,
we  solve  the system
$$
c_1 6 + c_2 0 + c_3 6 + c_4 13 = 1 \bmod 15
$$
to obtain $c=[c_1,c_2,c_3,c_4] = [-1,0,-1,1]$. Thus, $C_1$ and $W_1$ are
$$
\left[ \begin {array}{cccccccc} 1&&&&&&&\\ 
&1&&&&&&\\ &&1&&&&&
\\ &&&1&&&-1&\\ &&&&1&&0
&\\ &&&&&1&-1&\\ &&&&&
&1&\\ &&&&&&&1\end {array} \right]  
~\mbox{and}~ \left[ \begin {array}{cccccccc} 1&&&&&&&\\ 
&1&&&&&&\\ &&1&&&&&
\\ &&&1&&&&\\ &&&&1&&
&\\ &&&&&1&&\\ &&15&-6&0
&-6&1& \\ &&&&&&&1\end {array} \right] 
$$
respectively. After updating $\bar{U} = \rowmod(\bar{U} C_1 W_1,S)$ we have
$$
\bar{U} =
\left[ \begin {array}{cccccccc} \phantom{0}&\phantom{0}&0&0&0&0&0&0\\
&&0&0&0&0&2&2\\ &&0&3&0&8&4&2 \\ 
&&0&63&0&28&21&4\end {array} \right].
$$
Now we move on to iteration $j=2$. We have
$$
M S^* U = 
\left [ \begin{array}{cccccccc} \phantom{0} & \phantom{0} & \ast & \ast & 
\ast & \ast & \ast & \ast \\
 & & 0 & 9t_2 & 0t_2 & 4t_2 & \ast & \ast \\
 & &  &  & & & t_3 & \ast \\
 & & &  &  & & & t_4  \end{array} \right ]
$$
with $[a_1,a_2,a_3,a_4] = [0,9,0,4]$. Working over $\Z/(15)$
we solve the equation
$$
c_1 0 + c_2 9 + c_3 0 + c_4 4 = 1 \bmod 15
$$
to obtain $[c_1,c_2,c_3,c_4] = [0,1,0,-2]$. Thus, $C_2$ and $W_2$
are 
$$
\left [ \begin{array}{cccccccc} 1 & & & & & & & \\
 & 1 & & & & & &  \\
 & & 1 &  & & 0 & & \\
 & &   &1  & &1 & & \\
 & &   &   &1 &0 & & \\
 & &   &   & & -2 & & \\
 & &   &   & &  &1 & \\
 & &   &   & &  & & 1
\end{array} \right ] 
~\mbox{and}~
\left [ \begin{array}{cccccccc} 1 & & & & & & & \\
 & 1 & & & & & &  \\
 & & 1 &  & & & & \\
 & &   &1  & & & & \\
 & &   &   &1 & & & \\
 &15&0 &-9 &0 & 1 & & \\
 & &   &   & &  &1 & \\
 & &   &   & &  & & 1
\end{array} \right ] 
$$
respectively. 
After updating $\bar{U} = \rowmod(\bar{U} C_2 W_2,S)$ we have
$$
\bar{U} = \left[ \begin {array}{cccccccc} \phantom{0}&0&0&0&0&0&0&0\\ 
&0&0&0&0&0&2&2\\ &0&0&0&0&2&4&2
\\ &0&0&0&0&7&21&4\end {array} \right].
$$
Since $t_1=105$ implies $C_3=W_3=I_{2n}$, we can stop.
If we let $U$ be the submatrix of $\bar{U}$ comprised of the last
$n$ columns, then $MS^* U$ will be in Howell form.
\end{example}

\section{Scaled matrix vector product} \mylabel{sec:dotprod}

In order to obtain our softly cubic complexity, we need to show
that the key step in our special Howell triangulation algorithm
(Figure~\ref{fig:sht}) and in deducing the Hermite from Howell form
(Figure~\ref{fig:hvh}) can be computed efficiently. We do this by
giving an algorithm for the scaled matrix$\times$vector product problem
shown in Figure~\ref{fig:sdp2}.

\begin{figure}[H]
\centering
\fbox{
\begin{minipage}{.95\textwidth}
\texttt{ScaledMatVecProd}$(M,S,u,h,p)$\\
{\bf Input:}
\begin{itemize}
\item[(i)] A nonsingular Smith form  $S = \diag(s_1,\ldots,s_n )\in
\Z^{n \times n}$.\\
Note: Let $s := s_n$ and $S^* := sS^{-1}$.
\item[(ii)] $M
\in \Z^{n\times n}$ such that $M=\colmod(M,S)$.
\item[(iii)] $u \in \Z^{n \times 1}$ such that $u = \rowmod(u,S)$.
\item[(iv)] A divisor $h \in \Z_{\geq 1}$ of $s$ such that $(s/h)^{-1}M{S}^* u$ is over $\Z$.
\item[(v)] An odd prime $p$ such that $p \perp s$ and $\log p \in
\Theta(\loglog \det S)$.
\end{itemize}
{\bf Output:}
\begin{itemize}
\item[] $v = (v_i)_{1\leq i\leq n} \in \Z/(h)^{n\times 1}$ such
that 
$$
\frac{s}{h} 
\overbrace{
\left [ \begin{array}{c} v_{1} \\
 \vdots \\
v_n
\end{array} \right]}^{\textstyle v}  \equiv
\overbrace{
\left [\begin{array}{ccc}
m_{11} & \cdots & m_{1n}\\
\vdots &  \ddots & \vdots\\
m_{n1} & \cdots & m_{nn} 
 \end{array}
 \right ]}^{\textstyle M}
\overbrace{
\left [ \begin{array}{ccc}
\frac{s}{s_1} & & \\
  & \ddots & \\
  & & \frac{s}{s_n} \end{array} \right ]}^{\textstyle {S}^*}
\overbrace{
\left [ \begin{array}{c} u_{1} \\
\vdots \\
u_{n} \end{array} \right ]}^{\textstyle u}
\bmod s.
$$
\end{itemize}
\end{minipage}}
\caption{Problem \texttt{ScaledMatVecProd} \label{fig:sdp2}}
\end{figure}

From Lemma~\ref{remarkM}, we know that the sum of the bitlengths of the nontrivial columns of $M$
is bounded by $O(\log \det S)$.
Since ${S}^*u$ has entries reduced modulo $s$, Lemma~\ref{lem:AB} shows
that the matrix$\times$vector product $\Rem( M ({S}^*u), s)$ can be computed in 
\begin{equation}\label{eq:ideal}
O(n(\log \det S)(\log s))
\end{equation}
bit operations.  Dividing $\Rem(M({S}^*u),S)$ by $s/h$
gives the output
vector $v$.  

However, the cost estimate in~(\ref{eq:ideal}) is too high for
our purposes. Ideally, we would like to replace the $\log s$ factor 
in~(\ref{eq:ideal}) with $\log h$.  Instead, we are able to
obtain the following slightly weaker result.

\begin{theorem}\mylabel{thm:smvp}
Problem \texttt{ScaledMatVecProd}$(M,S,u,h,p)$ can be solved in 
\begin{equation}  \label{eq:cost}
O(n (\log \det S )(\log h + \loglog \det S) + (\log \det S)^2)
\end{equation}
bit operations.
\end{theorem}

In order to simplify the presentation of the algorithm, 
let $$m := \left [ \begin{array}{ccc} m_1 & \cdots & m_n \end{array} 
\right ]$$ denote a row of $M$.  Our goal then is to compute
a scalar $v \in \Z$ such that
\begin{equation}
\mylabel{eqmsu2}
\frac{s}{h} v \equiv
\overbrace{
\left [ \begin{array}{ccc} m_{1} & \cdots & m_{n} \end{array} \right ]
\left [ \begin{array}{ccc}
\frac{s}{s_1} & &  \\
  & \ddots & \\
  & & \frac{s}{s_n} \end{array} \right ]
\left [ \begin{array}{c} u_1 \\ 
\vdots \\
u_n \end{array} \right ]}^{\textstyle A} \bmod s.
\end{equation}
Afterwards,
we simply replace the row vector $m$
in~(\ref{eqmsu2}) with the matrix $M$.

We begin with a high level description of the algorithm.
The right hand side of~(\ref{eqmsu2}), if computed over
$\Z$ without taking $\bmod s$, is given by
\begin{equation} \mylabel{eq:A}
A=\sum_{i=1}^n\frac{s}{s_i}m_iu_i.
\end{equation}
An {\it a priori} magnitude bound is $A \in O(s^2 \log \det S)$.  
The formulation in~(\ref{eq:A}) highlights --- since we only
require an integer congruent to $A\bmod s$ --- that
the products $m_i u_i$ can be computed
modulo $s_i$ since they are scaled by $s/s_i$.  
In Subsection~\ref{ssec:pred1}, we show how to replace the scalar products $m_i u_i$ with dot products
that give an integer congruent to $m_i u_i$ modulo $s_i$.
This leads to a formula $D\equiv A \bmod s$ but with 
magnitude bound $D \in O(sh(\log \det S)^2)$.
Then in Subsection~\ref{ssec:pred2}, we show how to exploit the fact that
$(s/h)$ is a divisor of $D$, that is, 
$(h/s) D \in O(h^2(\log \det S)^2)$.

\subsection{Precision reduction via partial linearization} \mylabel{ssec:pred1}

Let $X \in \Z_{>1}$ be a positive radix
and, for a nonnegative integer $k$, define
$$\vec{X}^{(k)} := \left [ \begin{array}{c} X^0 \\
X^1 \\
\vdots\\
X^{k-1} \end{array} \right ] \in \Z^{k \times 1}.
$$
For $1\leq i \leq n$,  we let $\vec{m}_{i} \in \Z_{\geq 0}^{1 \times k_i}$ be the 
unique vector of coefficients of the $X$-adic expansion of $m_i$, that is, 
$||\vec{m}_i|| < X$ and $m_{i} = \vec{m}_{i} \cdot  \vec{X}^{(k_i)}$, where
$$k_i := \left \lceil \frac{\log s_i}{\log X} \right \rceil.$$
We can then rewrite the formula for $A$ in~(\ref{eq:A}) as
\begin{eqnarray}
\nonumber A & =  & \sum_{i=1}^n \frac{s}{s_i} \vec{m_i} \vec{X}^{(k_i)}u_i \\
\mylabel{eq:A2}  & = & \overbrace{\left [ \begin{array}{ccc} \vec{m}_1 & \cdots & \vec{m}_n \end{array}\right ]}^{\textstyle \vec{m}}
\left [ \begin{array}{ccc} \frac{s}{s_1} I_{k_1} & & \\
 & \ddots & \\
& & \frac{s}{s_n} I_{k_n} \end{array} \right ]
\left [ \begin{array}{c} u_1 \vec{X}^{(k_1)} \\
\vdots \\
u_n \vec{X}^{(k_n)} \end{array} \right ].
\end{eqnarray}
\begin{example} \mylabel{examp:lin1}
Let $m =[ 9, 7926]$, $u=[1012,8057]^t$ and $X=10$. Then, $A=mu$ can be
computed as
$$
A =  
\overbrace{\left[ \begin {array}{c|cccc} 9 & 6&2&9&7\end {array} \right]}^{\textstyle [\vec{m}_1, \vec{m}_2]} 
\overbrace{\left[ \begin {array}{c} 1012 \\\hline 8057\\ 80570
\\ 805700\\ 8057000\end {array}
\right]}^{\textstyle [u_1 \vec{X}^{(1)}, u_2 \vec{X}^{(4)}]^t}  = 63868890.
$$
\end{example}

For the components of $\vec{m}$, we will often  separately consider cases $k_i=1$ and 
$k_i>1$. Note that, in the latter case, $k_i>1$ implies $(\log s_i)/(\log X)~>~ 1$,
and hence we have the upper bound 
\begin{eqnarray}
k_i & =  &  \left \lceil \frac{\log s_i}{\log X} \right \rceil 
  \leq   1 + \frac{\log s_i}{\log X} 
  \leq  \frac{2 \log s_i}{\log X}. \mylabel{eqn:bndki}
\end{eqnarray}

\begin{lemma} \mylabel{lem:sizevecm} The sum of the bitlengths of the entries of $\vec{m}$ is bounded
by $O(\log \det S)$.
\end{lemma}
\begin{proof}
If $k_i=1$ then 
$\vec{m}_i$ consists of a single entry bounded in magnitude by $s_i>1$.
The sum of the bitlengths of all such entries of $\vec{m}$ is bounded by
$$
\sum_{\substack{i=1 \\k_i=1}}^n \lg s_i \leq  \sum_{\substack{i=1 \\k_i=1}}^n
 (1 + \log s_i) \leq   \sum_{\substack{i=1 \\k_i=1}}^n 
(2\log s_i) \leq  2 \log \det S.
$$
If $k_i>1$ then $\vec{m}_i$ contains $k_i$ entries with magnitude
bounded by $X$, and thus the sum of the bitlength
of entries in $\vec{m}_i$ is bounded by
\begin{equation}
k_i \lg X  \leq \left ( \frac{2 \log s_i}{\log X}\right ) (1 + \log X) 
 \leq  \left (\frac{2\log s_i}{\log X} \right )(2 \log X)
  \leq  4 (\log s_i),  \mylabel{dk4}
\end{equation}
with the first inequality coming from bound~(\ref{eqn:bndki}). The sum of
the right hand side of~(\ref{dk4}) over all $i$ with $k_i>1$ is thus also $O(\log \det S)$.
\end{proof}
Now we return to the reformulation of $A$ shown in~(\ref{eq:A2}).
Since we only require $A \bmod s$, we can preemptively reduce
the column vector in~(\ref{eq:A2}) by defining
$\vec{u}_i := \Rem(u_i \vec{X}^{(k_i)},s_i)$ for $1\leq i \leq n$.
Then
\begin{equation}\label{defB}
 D :=  \sum_{i=1}^n \frac{s}{s_i} \vec{m}_i \vec{u}_i
\end{equation}
is congruent to $A \bmod s$.  
\begin{example} \mylabel{examp:lin2}
Let $m =[ 9, 7926]$, $u=[1012,8057]^t$ and $X=10$ be as in Example~\ref{examp:lin1} and set $s=10000$. Then, 
$$
D =  
\overbrace{ \left[ \begin {array}{c|cccc} 9& 6&2&9&7\end {array} \right] }^{\textstyle \vec{m} }
\overbrace{\left[ \begin {array}{c} 1012 \\\hline \noalign{\medskip} 8057\\ \noalign{\medskip}570
\\ \noalign{\medskip}5700\\ \noalign{\medskip}7000\end {array}
 \right]}^{\textstyle \vec{u}}  = 158890
$$
is congruent modulo $s$ to $A=mu$.
\end{example}
Our first lemma derives a bound on the
magnitude of $D$.
\begin{lemma} 
\mylabel{lem:u}
Let $D$ be defined  as in~(\ref{defB}).  Then 
$D < 2sX \log \det S$.
\end{lemma}

\begin{proof}
From~(\ref{defB}), 
we see that
\begin{equation*} 
D = \left [ \begin{array}{ccc} \vec{m}_1 & \cdots & \vec{m}_n 
\end{array} \right ]
\left [ \begin{array}{c} (s/s_1) \vec{u}_1 \\
\vdots \\  (s/s_n) \vec{u}_n \end{array} \right ]
\end{equation*}
is a dot product of length $\sum_{i=1}^n k_i$, where the row vector
has entries from $[0,X)$, and the column vector
has entries from $[0,s)$.  This implies 
$D < s X \sum_{i=1}^n k_i.$
We can then bound this by
\begin{eqnarray*}
D & < & sX\sum_{i=1}^n k_i\\
&= & sX\sum_{i=1}^n \lceil \log s_i/\log X\rceil \\
& \leq & sX\sum_{\substack{i=1\\s_i\neq 1}}^n (1 + \log s_i) \\
& \leq & sX\sum_{\substack{i=1\\s_i\neq 1}}^n (2\log s_i) \\
& \leq & 2sX\log\det S.
\end{eqnarray*}
\end{proof}

Our next lemma bounds the cost of computing vector $\vec{u}_i$.  Note
that the lemma holds independently of the choice of $X$ (e.g., $X=2$ is valid).

\begin{lemma} \mylabel{lem:costui}
The vectors $\vec{u}_i$, for $1\leq i\leq n$, can be computed
in $O((\log \det S)^2)$ bit operations.
\end{lemma}
\begin{proof}
First consider the cost for a fixed $i$. If $k_i = 0$, then $\vec{u}_i
\in \Z^{0 \times 1}$, and there is no computation needed.
Similarly, if $k_i = 1$, then $\vec{u}_i = \left [ \begin{array}{c} u_i 
\end{array} \right ]$.
This leaves us with the case $k_i > 1$.  Let
$$  \left [ \begin{array}{c} a_1 \\
 a_2 \\
 \vdots \\
 a_{k_i} \end{array} \right ] := \vec{u}_i
= \left [ \begin{array}{c} \Rem(u_iX^0,s_i) \\
\Rem(u_iX^1,s_i) \\
\vdots \\
\Rem(u_iX^{k_i-1},s_i) \end{array} \right ]
$$
be our target vector. We can compute the $a_i$ using
a Horner scheme by:\\

\noindent
$a_1 := u_i$\\
\For $k=2$ \To $k_i$ \Do\\
\ind{1} $a_k := \Rem(a_{k-1}X,s_i)$\\
\Od\\

\noindent
By Lemma~\ref{lem:modmult}, there exists a constant $c$
such that the cost of one iteration of
the loop is bounded by $c(\log X)(\log s_i)$.
Since the loop iterates $k_i-1 < k_i$ times, the total cost
to compute $\vec{u}_i$ is $ck_i(\log X)(\log s_i)$.
The total cost to compute all $\vec{u}_i$ is then
\begin{eqnarray*}
\sum_{\substack{i=1\\k_i>1}}^n ck_i(\log X)(\log s_i)
& \leq &  c \sum_{\substack{i=1\\k_i>1}}^n \left (\frac{2\log s_i}{\log X}\right )
                    (\log X)(\log s_i)  \\
& \leq & 2c \sum_{\substack{i=1\\k_i>1}}^n (\log s_i)^2 \\
 & \leq  & 2c(\log \det S)^2,
\end{eqnarray*}
where the first inequality comes from~(\ref{eqn:bndki}).
\end{proof}

\subsection{Precision reduction via modular computation} \mylabel{ssec:pred2}
As shown in the proof of Lemma~\ref{lem:u}, we have
\begin{equation}  \label{eq:dp1}
D = \left [ \begin{array}{ccc} \vec{m}_1 &  \cdots & \vec{m}_n 
\end{array} \right ]
\left [ \begin{array}{c} (s/s_1) \vec{u}_1 \\
\vdots \\  (s/s_n) \vec{u}_n \end{array} \right ].\end{equation}
In order to reduce the precision of computing this dot product,
we can exploit the fact that $D$ has a known divisor $s/h$,
that is, $(h/s)D \in \Z$.  Multiplying~(\ref{eq:dp1}) by $(h/s)$
gives
\begin{equation} \label{lem:dp2}
(h/s) D = h \left [ \begin{array}{ccc} \vec{m}_1 &  \cdots & \vec{m}_n 
\end{array} \right ]
\left [ \begin{array}{c} (1/s_1) \vec{u}_1 \\
\vdots \\  (1/s_n) \vec{u}_n \end{array} \right ].\end{equation}
Lemma~\ref{lem:u} gives $D < 2sX\log \det S$ and hence $(h/s)D
< 2hX\log \det S$.  The idea now is to choose a modulus $Y \in
\Z_{>0}$ that is relatively prime to $s$ and satisfies $Y \geq 2hX\log \det S$.
Then, $(h/s)D < Y$.  Since  any integer $a$ that satisfies
$0 \leq a < Y$ gives $\Rem(a,Y) = a$, 
we can compute $(h/s)D$ by working modulo $Y$.
To this end, let
$\vec{w}_i = \Rem((1/s_i) \vec{u}_i,Y)$ for $1\leq i \leq n$.
Then,
\begin{equation}
\label{Bconst} (h/s)D  =  
  \Rem \left ( h  \left [ \begin{array}{cccc} \vec{m}_1 & 
\vec{m}_2 & \cdots & \vec{m}_n \end{array} \right ]
\left [ \begin{array}{c} \vec{w}_1 \\
\vec{w}_2 \\
\vdots \\  \vec{w}_n \end{array} \right ] , Y \right ) .
\end{equation}
In order to obtain a good complexity for computing the $\vec{w}_i$ vectors
from the $\vec{u}_i$ vectors, the moduli $X$ and $Y$ need to be
well chosen.

\begin{lemma} \mylabel{lem:w} If
 $X \in \Z_{>0}$ is the smallest power of 2 such that $X > 2h\log \det S$, and
$Y \in \Z_{>0}$ is the smallest power of $p$ such
that $Y > X^2$,
then 
\begin{itemize}
\item[(i)] $Y > 2hX\log \det S$,
\item[(ii)] $\log Y \in O(\log X)$, and
\item[(iii)] the vectors $\vec{w}_i$ for $1\leq i\leq n$ can be computed
from the vectors $\vec{u}_{i}$ in time $O((\log \det S)^2)$.
\end{itemize}
\end{lemma}

\begin{proof}
Part (i) follows by substituting $X > 2h\log \det S$ for one of the
factors of $X$ in the inequality  $Y > X^2$.  Part (ii) follows from
the choice of $Y$ as the smallest power of $p$, where $\log p \in
O(\loglog \det S)$ as per the problem specification.

For part (iii), we first  precompute $\bar{s}_i :=
\Rem(1/s_i,Y)$ for all $1~\leq~i~\leq~n.$  Note that $\bar{s}_i$ can
be computed by using the extended euclidean algorithm with input $(s_i,Y)$.
Thus, there exists a constant $c$ such that $\bar{s}_i$ can
be computed in time $c(\lg s_i)(\lg Y)$.  The total cost of computing
all the $\bar{s}_i$ is then bounded by
\begin{eqnarray}
\sum_{\substack{i=1\\s_i\neq 1}}^n c(\lg s_i)(\lg Y) & \leq & 
 c \sum_{\substack{i=1\\s_i\neq 1}}^n (1+ \log s_i)(1+ \log Y))\nonumber \\ 
 & \leq &  c \sum_{\substack{i=1\\s_i\neq 1}}^n (2 \log s_i)(2\log Y) \nonumber\\
 & \in & O((\log \det S)(\log Y)).
\label{eq:targ1}
\end{eqnarray}
The bound~(\ref{eq:targ1}) is within our target cost since $\log Y \in O(\log h + \loglog\det S)$,
which is bounded by $O(\log \det S)$ using the fact that $h \mid \det S$.

Since $\bar{s}_i < Y$ and $||\vec{u}_i|| < s_i$, 
it follows from Lemma~\ref{lem:modmult} that there exists a constant $c'$ such that
the cost of computing  $\vec{w}_i := \Rem(\bar{s}_i\vec{u}_i,Y) \in \Z^{k_i \times 1}$ 
is bounded by $c' k_i (\log s_i) (\log Y)$.
To bound the cost of computing all the $\vec{w}_i$ we consider separately
the case $k_i=1$ and $k_i > 1$.   For the case $k_i=1$ we obtain a total cost of
$$\sum_{\substack{i=1\\k_i=1}}^n  c' (\log s_i)(\log Y) \in O((\log \det S)(\log Y)),$$
which we have already seen to be within our cost bound.
For the case $k_i >1$ we obtain a total cost of
\begin{eqnarray*}
\sum_{\substack{i=1\\k_i>1}}^n  c' k_i (\log s_i)(\log Y) 
 & \leq &  c' \sum_{\substack{i=1\\k_i>1}}^n  \left ( \frac{(2\log s_i}{\log X} \right ) (\log s_i)(\log Y)\mbox{~~~~~~(\ref{eqn:bndki})}\\
 & \leq & \left ( \frac{2c'\log Y}{\log X} \right ) \sum_{\substack{i=1\\k_i>1}}^n  (\log s_i)^2 \\
 & \leq & 
  O((\log \det S)^2).
\end{eqnarray*}
The last inequality uses the fact that $\log Y \in O(\log X)$.
\end{proof}

\subsection{Proof of Theorem~\ref{thm:smvp}}

We first choose dual moduli $X$ and $Y$ as specified in Lemma~\ref{lem:w}.
Construct the partial linearization
$$\vec{M} = \left [ \begin{array}{ccc} \vec{m}_{11} & \cdots & \vec{m}_{1n} \\
\vdots & \ddots & \vdots \\
\vec{m}_{n1} & \cdots & \vec{m}_{nn} \end{array} \right ] \in \Z^{n
\times (k_1+\cdots + k_n)}$$ by replacing column $i$ of $M$ with
the $n \times k_i$ matrix containing the coefficients of its $X$-adic
expansion, for $1\leq i\leq n$.  Since $X$ is a power of 2, the
construction of $\vec{M}$ can be done in  time linear in the size
of $M$, thus in $O(n\log \det S)$ bit operations.

By Lemmas~\ref{lem:costui} and~\ref{lem:w}, we can compute in $O((\log
\det S)^2)$ time a vector 
\begin{equation} \mylabel{vecw} \vec{w}  = \left [\begin{array}{c} \vec{w}_1 \\ \vdots \\ \vec{w}_n 
\end{array} \right ] \in \Z/(Y)^{(k_1+\cdots +k_n) \times
1}\end{equation} such that our target vector $v$ is then given by
$v = \Rem(\Rem(h \vec{M} \vec{w},Y), h)$. We can thus compute
$v$ in three steps:
\begin{enumerate}
\item $a := \Rem(\vec{M} \vec{w},Y)$
\item  $b := \Rem(ha,Y)$
\item $v := \Rem(b,h)$ 
\end{enumerate}
By Lemmas~\ref{lem:sizevecm} and~\ref{lem:AB}, Step~1 can
be done in $O(n(\log \det S)(\log Y))$ bit operations.
By Lemma~\ref{lem:modmult}, Step~2 has cost $O(n(\log h)(\log Y))$ which,
since $h \mid \det S$, is bounded by $O(n(\log \det S)(\log Y))$.
Similarly, Step~3 computes $n$ division with remainder operations involving
the divisor $h$ and a dividend bounded in magnitude by $Y$, a step which also 
has cost $O(n(\log h)(\log Y))$.
This shows that once $\vec{w}$ is precomputed, computing the target vector $v$
can be done in time $O(n(\log \det S)(\log Y))$. Finally, by the definition of
$Y$ we have that $\log Y \in O(\log h + \loglog \det S)$.

\begin{remark} 
For clarity, Subsections~\ref{ssec:pred1} and~\ref{ssec:pred2} have
explained how to construct the vector $\vec{w}$ in~(\ref{vecw}) in
two steps: (a) first construct the vectors $\vec{u}_i \in
\Z/(s_i)^{k_i \times 1}$, $1\leq i\leq n$, as in Lemma~\ref{lem:costui};
(b) then use Lemma~\ref{lem:w} to construct the $\vec{w}_i \in
\Z/(Y)^{k_i \times 1}$ from $\vec{u}_i$.   An issue with producing
$\vec{u}$ explicitly is that it may require $\Omega((\log \det S)^2)$
bits to represent.  For this reason, each
of the $k_1+\cdots +k_n$ entries of $\vec{u}$ should be produced
one by one and then used to produce the corresponding entry of $\vec{w}$,
thus avoiding the need to store $\vec{u}$ explicitly.
With this adjustment, the intermediate space requirement of the algorithm
remains bounded by $O(n\log \det S)$ bits.
\end{remark}

\section{The Hermite form algorithm}\mylabel{sec:11}

At this point, we have developed all of the components for our
algorithm that computes the Hermite form $H$ of a nonsingular integer
matrix $A$.

Before we proceed with our main result, we note that all the
algorithms that have been given in Sections~\ref{sec:5}-\ref{sec:dotprod}
work with a reduced Smith massager $M$ and a Smith form $S$, and their cost
estimates depend on the dimension $n$ and $\log \det S$. For the
Hermite form algorithm, we would like to bound the cost in terms
of $n$ and $\log \|A\|$. Since $S$ will be the Smith form of $A$,
by Hadamard's bound, we have that

\begin{equation} \mylabel{eq:hadbound}
\log \det S = \log |\det A|\leq n\log\left(n^{1/2}\|A\|\right).
\end{equation}
Using (\ref{eq:hadbound}), the cost estimate $O(n(\log \det S)^2)$ from Theorem~\ref{thm:diag}, directly translates to
\begin{equation} \mylabel{eq:cost21simple}
O(n^3(\log n+\log \|A\|)^2).
\end{equation}
Similarly, in a slightly less trivial way, the cost estimate
\[O(n(\log \det S)^2 + n^2(\log \det S) (\loglog \det S))\]
from Theorems~\ref{thm:sht} and~\ref{thm:hvh} is also bounded by (\ref{eq:cost21simple}). The first part is the same as before, and for the second part
\begin{align*}
O(n^2(\log \det S) (\loglog \det S)) &\subseteq O(n^3(\log(n\|A\|))(\log(n\log(n\|A\|)))) \\
&\subseteq O(n^3(\log(n\|A\|))(\log n+ \log\log(n\|A\|))) \\
&\subseteq O(n^3(\log n+\log \|A\|)^2),
\end{align*}
since $O(\log n+ \log\log(n\|A\|))\subseteq O(\log n+\log \|A\|)$.

The following theorem is the main result of the article.

\begin{theorem} \mylabel{thm:hermstandard}
There exists a Las Vegas randomized algorithm that computes the
Hermite form $H\in\Znn$ of a nonsingular integer matrix $A\in\Znn$.
The algorithm uses standard integer and matrix multiplication and has cost
$O(n^3 (\log n + \log ||A||)^2(\log n)^2)$ bit operations.
\end{theorem}

\begin{proof}
The algorithm proceeds in four steps.

\begin{enumerate}
\item $M,S,p := \texttt{SmithMassager}(A)$
\item $h_1,\ldots,h_n := \texttt{HermiteDiagonals}(A,M,S)$
\item $U := \texttt{SpecialHowellTransform}(A,M,S,[h_1,\ldots,h_n], p)$
\item $H := \texttt{HermiteViaHowell}(A,M,S,U,[h_1,\ldots,h_n], p)$
\end{enumerate}

Step~1 uses the Las Vegas algorithm of \citet{BirmpilisLabahnStorjohann20,
BirmpilisLabahnStorjohann21}, restated in Theorem~\ref{thm:sm}, to
compute the Smith form $S$ and a reduced Smith massager $M$ of $A$.  The
cost is as stated in the current theorem.  Note that computing $M$ and $S$
is the only randomized component of the Hermite form algorithm.
The Smith massager algorithm also returns a prime $p$ such that
$p\perp\det S$. The prime is used in the
\texttt{ScaledMatVecProd} procedure in the algorithms used in Steps~3
and~4.

Step~2 exploits the fact that $M$ is maintained column modulo $S$
and computes the diagonal entries of $H$. By Theorem~\ref{thm:diag}
and Hadamard's bound 
this is done with 
\begin{equation} O(n^3(\log n+\log \|A\|)^2) \mylabel{hcost} \end{equation}
bit operations.

Step~3 computes a matrix $U\in\Znn$ such that $T=MS^* U$ is
right equivalent modulo $s$ to a Howell form of $MS^*(1/s)$, where $s$
is the largest invariant factor in $S$ and $S^* = sS^{-1}$. 
By Theorem~\ref{thm:sht} and
Hadamard's bound, the time complexity of Step~3 simplifies
to~(\ref{hcost}).

Finally, Step~4 computes the Hermite denominator $H$ of $T(1/s)$.
By Theorem~\ref{thm:hvh}, the cost of Step~4 is also~(\ref{hcost}).

To see correctness, note that by Definition~\ref{def:sm} the Hermite
denominator of $MS^{-1} = MS^*(1/s)$ is the Hermite form of $A$.
Since $T$ is right equivalent to $MS^*$ over $\Z/(s)$, $T(1/s)$ has
the same Hermite denominator as $MS^*(1/s)$ (cf.\ Remark~\ref{rem:U}).
The matrix $H$ computed in Step~4 is thus the Hermite form of $A$.
\end{proof}

\section{Using fast integer multiplication} \mylabel{sec:fhermalg}

Our Hermite form algorithm is designed to have a softly cubic
complexity in the parameter $n$ in an environment that assumes
standard integer multiplication: the cost of multiplying together
two integers of bitlength $d$ is $O(d^2)$ bit operations.  If we
are in an environment where integer multiplication
has cost $O(d^{1+\epsilon})$ bit operations for some $0 < \epsilon
\leq 1$, we can give a variation of our Hermite form algorithm that
establishes the following result.

\begin{theorem} \mylabel{thm:fast} 
There exists a Las Vegas randomized algorithm that computes the
Hermite form $H \in \Z^{n \times n}$ of a nonsingular integer matrix
$A \in \Z^{n \times n}$ using
$O(n^{3+\epsilon}(\log ||A||)^{1+\epsilon})$
bit operations.  
\end{theorem}

Before proving the theorem, we give
three lemmas. Let $S = \diag(s_1,\ldots,s_n)$
be a nonsingular Smith form, and let $M\in\Znn$
satisfy $M = \colmod(M,S)$. Also, let $s := s_n$ and $S^* := sS^{-1}$.

Consider the update step $M := \colmod(H_j M,S)$ required in the proof
of Theorem~\ref{thm:hvh}. The dominant cost is to compute the outer
product of column $j$ of $H_j$ with row $j$ of $M$, keeping this
column reduced modulo $S$. Our first lemma shows that this can be done
efficiently. We also use the lemma in the transpose situation
to bound the cost of the update $\bar{U} := \rowmod(\bar{U}W_i,S)$
required in the proof of Theorem~\ref{thm:sht}.

\begin{lemma} \mylabel{lem:outfast}
Given
a $u \in \Z/(s)^{n \times 1}$, together with
an $m \in \Z^{1 \times n}$ such that $m= \colmod(m,S)$,
we can compute
$\colmod(u m, S)$ in 
$O(n (\log \det S)^{1+\epsilon})$
bit operations.
\end{lemma}
\begin{proof}
Let $m = \left [ \begin{array}{ccc} m_1 & \cdots & m_n \end{array}
\right ] \in \Z^{1 \times n}$. Then
$$\colmod(u m,S)  = \left[  \begin{array}{ccc} 
\bar{u}_1 &  \cdots & \bar{u}_n 
\end{array} \right ],$$
where $\bar{u}_i = \Rem(um_i,s_i) \in \Z/(s_i)^{n \times 1}$, $1\leq i\leq n$.
Note that if $s_i=1$ then $\bar{u}_i$ is necessarily the zero vector.
The $\bar{u}_{\ast}$ that are not necessarily zero 
can be computed using the following loop:\\

\noindent
$\bar{u} := u$\\
$\bar{u}_n := \Rem(\bar{u}m_n,s_n)$\\
\For $i$ \From $n-1$ \Downto 1 \Do\\
\ind{1} \If $s_i=1$ \Then \Break\/ \Fi\\
\ind{1} $\bar{u} := \Rem(\bar{u},s_i)$\\
\ind{1} $\bar{u}_i := \Rem(\bar{u}m_i,s_i)$\\
\Od\\

\noindent
The cost of computing $\bar{u}_n$ is bounded by 
$O(n(\log s)^{1+\epsilon})$
bit operations.
Since the operands at loop iteration $i$ have bitlength bounded by
$\lg s_{i+1}$, the cost at iteration $i$ is 
$O(n (\lg s_{i+1})^{1+\epsilon})$
bit operations. The total cost of the loop is thus
$O(n \sum_{i=2,s_i\neq 1}^n (\lg s_i)^{1+\epsilon})$.
Using the fact that $\sum_{i=2,s_i\neq 1} \lg s_i \in O(\log \det S)$, the 
total cost to compute the $\bar{u}_*$ is as stated in the lemma.
\end{proof}

Furthermore, consider the update step $\bar{U} := \rowmod(\bar{U} C_i,S)$ in the proof
of Theorem~\ref{thm:sht}. Since $C_i$ has at most one nontrivial column,
the dominant cost is to compute a matrix$\times$vector product, keeping
this row reduced modulo $S$.
The following corollary, applied to the transpose situation, shows
that this can be done efficiently.
The proof is analogous to the proof of Lemma~\ref{lem:outfast}.

\begin{corollary} \mylabel{cor:vecmat}
Given a $u \in \Z/(s)^{1 \times n}$, together with 
an $M \in \Z^{n \times n}$ such that $M = \colmod(M,S)$, 
we can compute $\colmod(uM,S)$ in 
$O(n(\log \det S)^{1+\epsilon})$ 
bit operations.
\end{corollary}

The following result will be used in place of \texttt{ScaledMatVecProd}.

\begin{lemma}[\protect{\citet[Lemma~4.11]{Storjohann10a}}] \mylabel{lem:inv}
Given an $M \in \Z^{n \times n}$ such that $M= \colmod(M,S)$,
together with a $U \in \Z^{n \times n}$ such that $U = \rowmod(U,S)$,
then any individual row or column of $\Rem(MS^* U,s)$
can be computed  using 
$O(n (\log \det S)^{1+\epsilon})$ 
bit operations.
\end{lemma}

We now prove Theorem~\ref{thm:fast}.

\begin{proof} (Of Theorem~\ref{thm:fast}).
We begin by~(i) computing the Smith form $S$ and a reduced Smith massager
$M$ of $A$, then~(ii) compute an integer matrix $U$ such that
$M(s_nS^{-1})U$ is right equivalent to a Howell form of $sA^{-1}$
over $\Z/(s)$, and finally~(iii) compute $H$ as the Hermite denominator
of $MS^{-1}U$.

\citet[Theorem~19]{BirmpilisLabahnStorjohann21} establish that
phase~(i) can be done within the time stated in Theorem~\ref{thm:fast}.

For phase~(ii), we adapt the algorithm, with $n$ iterations and three
steps per iteration, given in the proof of Theorem~\ref{thm:sht}.
In Step~1, use Lemma~\ref{lem:inv} to compute the required $n$
entries \begin{equation} \mylabel{nent}
\left [ \begin{array}{cccc}
t_{n-i}a_1 & \cdots & t_{n-i} a_{n-1} & t_{n-i} a_n 
\end{array} \right ] \in \Z/(s)^{1 \times n}.
\end{equation}
Since we are not given $t_{n-i}=s/h_{n-i}$ as input,
we compute it now as the gcd of entries of the $n$ elements in~(\ref{nent})
at a cost of 
\begin{equation} \mylabel{costh} 
O(n (\log s)^{1+\epsilon})
\end{equation} bit operations.
In Step~2, the update matrices $C_i$ and $W_i$ can be computed
in the time~(\ref{costh}) using an analog of Lemma~\ref{smlem}.
In Step~3, the update $\bar{U} := \rowmod(\bar{U}C_iW_i,S)$
is done in time~(\ref{costh}) using Lemmas~\ref{cor:vecmat} 
and~\ref{lem:outfast}.
Since there are $n$ iterations, and 
$\lg s \leq \lg \det S \in O(n(\log n + \log ||A||))$, 
the overall cost of phase~(ii) is as stated in the theorem.

Similar to phase~(ii), the overall cost bound for phase~(iii)
follows by adapting the two-step
algorithm in the proof of Theorem~\ref{thm:hvh}
by using Lemma~\ref{lem:inv} for Step~1, and Lemma~\ref{lem:outfast} for
Step~2.
\end{proof}

Finally, if we assume we are using a pseudo-linear algorithm
for integer multiplication, such as the $O(d\log d)$ algorithm
of \citet{HarveyvanderHoeven21}, we obtain the following corollary.

\begin{corollary}
There exists a Las Vegas randomized algorithm that computes the
Hermite form $H \in \Z^{n \times n}$ of a nonsingular integer matrix
$A \in \Z^{n \times n}$ using
$(n^3\log ||A||)^{1+o(1)}$
bit operations. This cost estimate assumes
the use of a pseudo-linear algorithm for integer multiplication.
\end{corollary}

\section{Conclusion and topics for future research}\label{sec:12}

We have given a Las Vegas randomized algorithm to compute the Hermite
form $H \in \Z^{n \times n}$ of a nonsingular matrix $A \in \Z^{n
\times n}$. The algorithm has worst-case expected running time
\begin{equation} \mylabel{eq:costalg} O(n^3 (\log n + \log ||A||)^2
(\log n)^2) \end{equation} bit operations using standard integer
and matrix multiplication.

The core tool used is the Smith massager which helps control the
size of intermediate results.  The $(\log n)^2$ factor
in~(\ref{eq:costalg}) is due to the first step of the algorithm,
which computes a Smith form $S$ and Smith massager $M$ of $A$.  This
first step is accomplished using the Las Vegas algorithm of
\citet[Theorem~19]{BirmpilisLabahnStorjohann21} which allows the
use of fast matrix multiplication, and shows that $S$ and $M$ can
be computed using an expected number of $O(n^{\omega}\, (\log n  +
\log ||A||)^2 (\log n)^2)$ bit operations assuming standard integer
multiplication.  Computing $M$ is also the only part of the Hermite
form algorithm that requires randomization.

Once $M$ is precomputed, the algorithm in this paper computes $H$
deterministically using a further  $O(n^3 (\log n  + \log ||A||)^2)$
bit operations.  The intermediate space requirement of the algorithm
to compute $H$ from $M$ is bounded by $O(n^2(\log n + \log ||A||))$
bits, which is the same as that required to write down $H$ in the
worst case.

We have also given a variant of our Hermite form algorithm that has
a worst case expected running time $(n^3\log ||A||)^{1+o(1)}$ bit
operations, assuming the use of a pseudo-linear algorithm for integer
multiplication.

Our Hermite form algorithms extend to the case of an input matrix
$A \in \Z^{m \times n}$ of full column rank $n$ and $m>n$.  Up to
a row permutation, and up to adding at most $n-1$ zero rows, we may
assume without loss of generality that
$$
A = \left [ \begin{array}{c} A_1 \\ A_2  \\ \vdots \\ A_k 
\end{array} \right ] \in \Z^{kn \times n},
$$
where each $A_{\ast}$ is $n \times n$, 
$A_1 \in \Z^{n \times n}$ is nonsingular, and $k = 
\lceil n/m \rceil$.
Initialize $H_1 := A_1$.
Compute, in succession for $i=2,3,\ldots,k$, 
the leading principal $n \times n$ submatrix 
$H_{i}$ of the Hermite form of the nonsingular matrix
\begin{equation} \mylabel{eqHH}
\left [ \begin{array}{cc} H_{i-1} & \\
A_i & I_n \end{array} \right ] \in \Z^{2n \times 2n}.
\end{equation}
Then $H_k \in \Z^{n \times n}$ is the leading principal $n\times n$ submatrix
of the Hermite form of $A$.  \citet[Theorem~27 and
Remark~34]{BirmpilisLabahnStorjohann21} show that computing the
Hermite form of~(\ref{eqHH}) reduces to that of computing
the Hermite form of a matrix of dimension bounded by $4n$ that
has entries with bitlength $O(\log n + \log ||A||)$. Computing
the Hermite form of an $A \in \Z^{m \times n}$ of rank $n$ can
thus be done in a Las Vegas fashion using 
an expected number of $O(mn^2(\log n + \log ||A||)^2(\log n)^2)$
bit operations using standard integer and matrix arithmetic, or an expected
number of $O(mn^2 \log ||A||)^{1+o(1)}$
bit operations using pseudo-linear integer multiplication.

In terms of future directions, a natural goal is to find an algorithm
to compute the Hermite form of a nonsingular integer matrices that
has cost $(n^\omega \log ||A|| )^{1+o(1)}$ bit operations.  In
addition, we would like to find a deterministic algorithm for the
Hermite form problem with the same complexity.

\bibliographystyle{plainnat}
\bibliography{short}

\newcommand{\SortNoop}[1]{}

\end{document}